\documentclass[12pt]{article}
\usepackage{euscript,amsmath,amssymb,latexsym,amsthm,enumerate,mathrsfs}
\usepackage[bookmarks=true]{hyperref}
\usepackage{cite}

\DeclareMathOperator{\sgn}{sgn}

\renewcommand*{\Re}{\mathop{\mathrm{Re}}\nolimits}

\theoremstyle{remark}
\newtheorem{remark}{Remark}

\theoremstyle{plain}
\newtheorem{theorem}{Theorem}
\newtheorem{lemma}{Lemma}
\newtheorem{proposition}{Proposition}
\newtheorem{corollary}{Corollary}

\begin{document}
\title{
\vspace{1cm} {\bf Localization properties of squeezed quantum states in nanoscale space domains}
}
\author{A.\,S.~Trushechkin$^{1,2}$ and I.\,V. Volovich$^{1}$\bigskip
 \\
{\it  $^{1}$Steklov Mathematical Institute of the Russian Academy of Sciences}
\\ {\it Gubkina St. 8, 119991 Moscow, Russia}\medskip\\
{\it  $^{2}$National Research Nuclear University ``MEPhI''}
\\ {\it Kashirskoe Highway 31, 115409 Moscow, Russia}\bigskip
\\ e-mail:\:\href{mailto:trushechkin@mi.ras.ru}{\texttt{trushechkin@mi.ras.ru}}, \: \href{mailto:volovich@mi.ras.ru}{\texttt{volovich@mi.ras.ru}}}

\date{}
\maketitle

\begin{abstract}

We construct families of squeezed quantum states on an interval (depending on boundary conditions, we interpret the interval as a circle or as the infinite square potential well) and obtain estimates of position and momentum dispersions for these states. A particular attention is paid to the possibility of proper localization of a particle in nanoscale space domains. One of the constructed family of squeezed states is based on the theta function. It is a generalization of the known coherent and squeezed states on the circle. Also we construct a family of squeezed states based on truncated Gaussian functions and a family of wave packets based on the discretization of an arbitrary continuous momentum probability distribution.

The problem of finiteness of the energy dispersion for the squeezed states in the infinite well is discussed. Finally, we perform the limit of large interval length and the semiclassical limit. 

As a supplementary general result, we show that an arbitrary physical quantity has a finite dispersion if and only if the wave function of a quantum system belongs to the domain of the corresponding self-adjoint operator. This can be regarded as a physical meaning of the domain of a self-adjoint operator.
\end{abstract}

\section{Introduction}

Coherent states on the real line are well-known in quantum mechanics; they were introduced by Schr\"{o}dinger \cite{Schroe-coher,Neumann}. The behaviour of a quantum system in such states is in a sense close to the behaviour of appropriate classical systems. A more general class of states is formed by squeezed states, which are obtained from coherent states by the dilation transformation \cite{Schleich}.

Coherent and squeezed states on manifolds and bounded space domains are studied as well, the definitions and properties of these states are still being discussed in the literature \cite{CN,Wei,KS, Per,TAG-book,V,G-book,CD-book}. In this paper, we are interested in one-dimensional systems like a quantum particle on a circle \cite{BG,KRCircle,GdOCircle,HM,KR2002,GCircle,CLTCircle} and in the infinite square potential well \cite{PTinf,GenIntel,Gazeau,HW}.

The aim of this work is to obtain estimates on position and momentum dispersions for such states, other words, to research their localization properties. A particular attention is paid to the possibility of proper localization of a particle in nanoscale space domains \cite{nano}. To deal with nanoscale systems, one should be able to localize quantum particles with high accuracy in an acceptable range of momenta. Recall that the possibility of such localization in an infinite volume is restricted by the Heisenberg uncertainty principle. It was established that for a finite volume the uncertainty principle needs essential modification (see the discussion in Subsection~\ref{SecUncert}). It is not obvious without additional analysis that in a bounded volume there exist quantum states for which quantum particles can be localized with an accuracy necessary for nanotechnological operations.

To this aim, we adopt the notion of squeezed states as states such that the uncertainty principle is saturated. However, there is no consensus about the ``right'' form of the uncertainty principle for bounded domains. By this reason, we demand squeezed states on an interval to saturate the usual uncertainty relation on the line asymptotically.

We construct families of squeezed quantum states on an interval and study their asymptotic behaviour. We obtain estimates for the position and momentum dispersions of a quantum particle on an interval in such states; these estimates can be applied, in particular, to nanoscale systems. We examine the localization properties of squeezed states on an interval.

The following text is organized as follows. In Section~\ref{SecProblem} we define the model and formulate the problem. Namely, in Subsection~\ref{SecModelDescr}, we introduce the operators of energy, position, and momentum for a quantum particle on a circle and in the infinite potential well. In Subsection~\ref{SecUncert}, we review some facts about the uncertainty relations and coherent and squeezed quantum states on the line. We remind the known fact that the usual Heisenberg uncertainty relations does not hold for a quantum particle on an interval.  In Subsection~\ref{SecProblemFormulation}, we give our definition of squeezed states on an interval.

In Section~\ref{SecSq}, we introduce various families of squeezed states on an interval. In Subsection~\ref{SecPsiGauss}, we introduce a family of squeezed states on an interval based on truncated Gaussian functions. The main result of this subsection is Theorem~\ref{TheoPsiGauss}, which gives estimates on mean values and dispersions of position and momentum of a particle in these states. In Subsection~\ref{SecPsiTheta}, we introduce a family of squeezed states on an interval based on the theta function. They generalize the known circular coherent and squeezed states introduced in \cite{BG,GdOCircle,KRCircle,KR2002,HM}. The main result is Theorem~\ref{TheoPsiTheta} with estimates of mean values and dispersions. In Subsection~\ref{SecArbitrary}, we introduce a family of quantum states on an interval based on the discretization of an arbitrary continuous momentum probability distribution (the construction is similar to that of \cite{GCircle}). The required estimates are given in  Theorem~\ref{TheoPackOtrExist}.

We obtain the following numerical estimates: on an interval of order 100~nm, there exist wave packets with a standard deviation of the position of order 0.1~nm and a standard deviation of the momentum of order  $10^{-24}$~kg$\cdot$ m/s. This is a result of Subsections~\ref{SecPsiGauss} and \ref{SecPsiTheta}. So, we prove that the proper localization of quantum particles in nanoscale space domains is possible. The corresponding numerical result of Subsection~\ref{SecArbitrary} is rougher: Theorem~\ref{TheoPackOtrExist} can guarantee the condition $\Delta x_\alpha\lesssim0.1$~nm only when $\Delta
p_\alpha\sim10^{-20}$~kg$\cdot$m/s. But this is still enough for the localization of a quantum particle in a nanoscale space domain. 

Section~\ref{SecFurther} is devoted to further related problems. We want a particle to be localized not only in the position and momentum spaces, but in the energy space as well. In Subsection~\ref{SecEnergy}, we discuss the energy localization of a quantum particle. We address an additional difficulty that arises when dealing with quantum wave packets in the infinite square well: in this case the momentum and energy operators do not commute.

As a supplementary general result, in Subsection~\ref{SecDisp}, we establish the relation between the finiteness of the dispersion of an arbitrary physical quantity and the domain of the corresponding self-adjoint operator.  This relation is not obvious: the ``mathematical'' questions concerning the domains of self-adjoint operators are often omitted in physical literature on quantum mechanics. Here we establish the physical meaning of the domain of a self-adjoint operator: an arbitrary physical quantity has a finite dispersion if and only if the wave function of a quantum system belongs to the domain of the corresponding self-adjoint operator. 

In Subsection~\ref{SecLim}, we perform the limit of the large interval length. We obtain the well-known squeezed quantum states on the line in this limit.  Also we also show that in the semiclassical limit both momentum and position dispersions of squeezed states on an interval vanish.

The most cumbersome calculations for the proofs of Theorems~\ref{TheoPsiGauss}--\ref{TheoPackOtrExist} are given in Appendixes.

\section{Formulation of the problem}\label{SecProblem}

\subsection{Quantum particle on an interval}\label{SecModelDescr}

A quantum particle on an interval $[-l,l]$ is associated with the Hilbert space $L_2(-l,l)$ (see \cite{Neumann}). A (pure) state of a particle is given by a unit vector from this space. If a particle moves freely between the ends of the interval, then the Hamiltonian on the subspace $C_0^\infty(-l,l)$ (of infinite differentiable functions with support in a subinterval of $[-l,l]$) is given by
$$\check{H}=-\frac{\hbar^2}{2m}\frac{d^2}{dx^2},$$
where $m>0$ is the mass of the particle and $\hbar$ is the Planck constant. The operator
$\check H$ with the domain $D(\check H)=C_0^\infty(-l,l)$
is symmetric but not self-adjoint. It has various self-adjoint extensions corresponding to different physical situations~\cite{ReedSimon} (see also \cite{Teach,VGT,VGT1,VGT2,VGT3}). Each self-adjoint extension is defined on functions from $AC^2(-l,l)$ and satisfy a linearly independent pair of boundary conditions of the form
\begin{equation}\label{EqGranUslGen}
A\psi(l)+B\psi(-l)+C\psi'(l)+D\psi'(-l)=0,\end{equation}
where $A,B,C,D\in\mathbb C$. Here $AC^2(-l,l)$ is the set of differentiable functions from $L_2(-l,l)$ whose derivatives
belong to $AC(-l,l)$, and $AC(-l,l)$ is the set of absolutely continuous functions whose derivatives
(which exist almost everywhere according to the properties of absolutely continuous functions) lie
in $L_2(-l,l)$. Each self-adjoint extension $\check H$ acts on its own domain as $(-\frac{\hbar^2}{2m}\frac{d^2}{dx^2})$. Note that not all (even linearly independent) pairs of boundary conditions of the form (\ref{EqGranUslGen}) correspond to self-adjoint extensions of the operator $\check H$. Necessary and sufficient conditions on the coefficients under which such a pair of boundary conditions induces a self-adjoint extension are given, for example, in \cite{Naimark} and \cite{VGT,VGT1,VGT2,VGT3}; however, these conditions are inessential in this study.

Let us notice two self-adjoint extensions of the operator $\check H$. For a particle in the infinite square potential well with rigid walls, the corresponding Hamiltonian
$$\hat{H_1}=-\frac{\hbar^2}{2m}\frac{d^2}{dx^2}$$
is defined on the domain
$$D(\hat H_1)=\{\psi\in AC^2(-l,l)|\,\psi(-l)=\psi(l)=0\}.$$

The eigenvalues and eigenfunctions of this operator can easily be found, and they are
well-known \cite{Davydov}:
$$E_n^{(1)}=\frac{\hbar^2}{2m}\left(\frac{\pi n}{2l}\right)^2,\quad
\psi_n^{(1)}=\frac{1}{\sqrt l}\sin\left(\frac{\pi n}{2l}(x-l)\right),\quad
n=1,2,\ldots.$$ It is easily seen that the probability current for the eigenfunctions of $\hat H_1$ vanishes at every point. We can express them as
$$\frac{1}{\sqrt l}\sin\left(\frac{\pi n}{2l}(x-l)\right)=\frac{1}{2i\sqrt l}[e^{\frac{i\pi n}{2l}(x-l)}-e^{-i\frac{\pi n}{2l}(x-l)}].$$
The densities of two counterpropagating waves $e^{\frac{i\pi n}{2l}(x-l)}$ and
$e^{-i\frac{\pi n}{2l}(x-l)}$ are the same; upon reflecting from the walls, these waves change the direction of propagation.  One can say that, upon reflection from the walls, they turn into each other.

For a particle on a circle (of length $2l$), the corresponding Hamiltonian
$$\hat{H_2}=-\frac{\hbar^2}{2m}\frac{d^2}{dx^2}$$
is defined on the domain
$$D(\hat H_2)=\{\psi\in AC^2(-l,l)|\,\psi(-l)=\psi(l),\,\psi'(-l)=\psi'(l)\}.$$
The (doubly degenerate) eigenvalues and eigenfunctions of the operator $\hat H_2$
can easily be found
as well:
$$E^{(2)}_n=\frac{\hbar^2}{2m}\left(\frac{\pi}{l}n\right)^2,\quad
\psi_n^{(2)}=\frac{1}{\sqrt{2l}}\,e^{i\frac{\pi}{l}nx},\,
\psi_{-n}^{(2)}=\frac{1}{\sqrt{2l}}\,e^{-i\frac{\pi}{l}nx},\quad
n=0,1,2,\ldots.$$ It is easily seen that the probability current for the eigenfunctions of $\hat H_2$ is constant and different from zero at every point. This means that such an eigenfunction corresponds to a stream of particles that move in the same direction: the part of the wave function that goes beyond the interval appears on the other side.

Define the momentum operator. Let a wave packet does not touch the endpoints of the interval and belongs to $C^\infty_0(-l,l)$. Then the momentum operator for this particle must be the same as for a particle on the line. That is, on the subspace $C^\infty_0(-l,l)$, the momentum operator has the form
\begin{equation*}%\label{EqMomOpPre}
\check p=-i\hbar\frac{d}{dx}.
\end{equation*} 
The operator $\check p$ with the domain $D(\check p)=C^\infty_0(-l,l)$ is also symmetric but not self-adjoint. All self-adjoint extensions of this operator are parametrized by real numbers
$\theta\in[0,2\pi)$ as follows \cite{ReedSimon}:
\begin{equation*}%\label{EqMomOp}
\hat p_\theta=-i\hbar\frac{d}{dx},\quad
D(\hat p_\theta)=\{\psi\in
AC(-l,l)|\,\psi(-l)=e^{i\theta}\psi(l)\}.
\end{equation*} 
As the momentum operator, we will consider $\hat p_0\equiv\hat p$. In fact, our further results can easily be carried over to the case of an arbitrary $\theta$.

The eigenvalues and eigenfunctions of the momentum operator $\hat
p$ are as follows:
\begin{equation*}%\label{EqImpSpecOtr}
p_k=\frac{\pi}{l}\hbar k,\quad
\varphi_k=\frac{1}{\sqrt{2l}}\,e^{i\frac{\pi}{l}kx}=\frac{1}{\sqrt{2l}}\,e^{i\frac{p_k}{\hbar}x},\quad
k=0,\pm1,\pm2,\ldots.\end{equation*}

Note that in the case of a particle in the infinite well, the energy and momentum operators do not commute, because, being operators with purely discrete spectra, they do not have a common set of eigenfunctions. In the case of a particle on a circle, these operators commute, just as in the well-known case of a particle on the line.

\begin{remark}
There are different views on whether the spectrum of the momentum of a particle in an infinite well is discrete or continuous, and there are different approaches to defining the corresponding operator. We define the momentum operator according to the standard formalism of quantum mechanics. In this case, the spectrum of the momentum turns out to be discrete.

However, one may argue that the momentum in this case should have a continuous spectrum, and the standard formalism does not fully comply with physics. This is justified as follows. The infinitely deep potential well is an idealization, an approximation of a potential well of very large but finite depth. In the case of a well of any finite depth one deals with the space $L_2(\mathbb R)$; i.e., the momentum spectrum is continuous. As the depth of the well tends to infinity, the momentum spectrum does not become discrete. Thus, the momentum spectrum in the infinite well should also be continuous.

From the other side, in this case one faces the problem of formally defining a self-adjoint momentum operator in the Hilbert space for the infinite well. An attempt to solve this problem was made in  \cite{Karwowski}, where one can also find a review of relevant literature.

In this paper, we apply the standard formalism of quantum mechanics, according to which the momentum operator is defined as above and its spectrum is discrete. In any case, this formalism can be used to describe a quantum particle on a circle.

The motion of a Bloch particle in a crystal under the influence of a magnetic field was considered in \cite{Nov}.
\end{remark}

As is usual in quantum mechanics, the position operator $\hat x$ is the multiplication of a function in the position representation by the variable $x$, $\hat x\psi(x)=x\psi(x)$, and the domain of $\hat
x$ is the whole $L_2(-l,l)$. In this Hilbert space, in contrast to the case of a quantum particle on the line, the position operator is bounded and defined on the whole space $L_2(-l,l)$. Therefore, in contrast to the operators of energy and momentum, it is self-adjoint just because it is symmetric.

\begin{remark}\label{RemPosCircle}
In the case of a particle on a circle it is more preferable to use the operator $e^{i\frac\pi l\hat x}$ instead of the operator $\hat x$, because the last is not periodic in $x$ (that is, $\hat x+2l\neq\hat x$, but $e^{i\frac\pi l(\hat x+2l)}=e^{i\frac\pi l\hat x}$). However, we will be interested in  wave functions localised near some point. We will denote this point as $x=0$. In this case, we can use the usual operator $\hat x$ since the wave function is essentially non-zero only in some neighbourhood of the point $x=0$. 
\end{remark}

\subsection{Uncertainty principle on the line and an interval}\label{SecUncert}

Heisenberg uncertainty relation for a particle on the line is well-known:

\begin{equation}\label{EqUncert}
\Delta x\Delta p\geq\frac{\hbar}{2},
\end{equation}
where $\Delta x$ and $\Delta p$ are, respectively, the position and momentum standard deviations of the particle in a state $\psi\in L_2(\mathbb R)$:
\begin{align}
\Delta x^2=\int_{\mathbb R}(x-\overline x)^2|\psi(x)|^2\,dx,\quad
&\Delta p^2=\int_{\mathbb
R}\overline{\psi(x)}(-i\hbar\frac{d}{dx}-\overline
p)^2\psi(x)\,dx,\\
\overline x=\int_{\mathbb R}x|\psi(x)|^2\,dx,\quad &\overline
p=\int_{\mathbb
R}\overline{\psi(x)}(-i\hbar)\frac{d}{dx}\psi(x)\,dx.\end{align}

The states that minimize this uncertainty relation are also well-known. These are Gaussian wave packets parametrized by three real numbers $x^*,p^*$, and $\alpha>0$:

\begin{equation}\label{EqCoher}
\psi(x)=\frac{1}{\sqrt[4]{2\pi\beta^2}}\,e^{-\frac{(x-x^*)^2}{4\beta^2}+i\frac{p^*(x-x^*)}{\hbar}}=
\frac{1}{\sqrt{2\pi\hbar}}\int_{-\infty}^{+\infty}\widetilde\psi(p)
e^{\frac{ipx}{\hbar}}\,dp,\end{equation}
$$\widetilde\psi(p)=\frac{1}{\sqrt[4]{2\pi\alpha^2}}\,e^{-\frac{(p-p^*)^2}{4\alpha^2}-i\frac{px^*}{\hbar}},$$
In this case $\bar x=x^*$, $\bar p=p^*$, $\Delta
x=\beta=\frac{\hbar}{2\alpha}$, and $\Delta p=\alpha$. Thus,
\begin{equation}\label{EqUncertMin}
\Delta x\Delta p=\frac{\hbar}{2}.
\end{equation} 
That is, Gaussian wave packets minimize the uncertainty relation. These states are called squeezed (or coherent, if the parameter $\alpha$ is fixed). One of the fields of application of these states is the classical approximation of quantum mechanics: since these states are the closest to classical states, a classical particle with arbitrary position $x^*\in\mathbb R$ and momentum $p^*\in\mathbb R$ is associated with a coherent state that minimizes the uncertainty relation and has the same average values $x^*$ and $p^*$ of the position and momentum, respectively.

The above relations for a Gaussian wave packet imply that $\Delta x$ can be made arbitrarily small at the cost of increasing
$\Delta p$ but keeping $\Delta p$ finite. On the other hand, one can make $\Delta p$ arbitrarily small at the cost of increasing $\Delta x$ but keeping $\Delta x$ finite. One can also choose $\Delta x$ and
$\Delta p$ such that both these quantities are small compared with macroscopic scales. Since
$\hbar\sim10^{-34}$ J$\cdot$s, we find, for example, that there exist wave packets satisfying the following estimates: $\Delta x\sim0.1$~nm and $\Delta
p\sim10^{-24}$ kg$\cdot$m/s.

We want to investigate analogous problems for a quantum particle on an interval. This case differs from the case of a particle on the line in the following essential aspects.

First, the momentum spectrum of the particle on an interval is discrete rather than continuous. This implies, as we will see, that there exist wave functions such that
\begin{equation}\label{EqUncertZero}
\Delta x\Delta p=0,\end{equation} 
For example, for the momentum eigenstates $\varphi_k=\frac1{2l}e^{i\frac\pi lkx}$, we have $\Delta p=0$, $\Delta x=\frac l{\sqrt 3}$ (see below) and, hence, (\ref{EqUncertZero}) holds.

Therefore, the well-known uncertainty relation (\ref{EqUncert}) is not valid on an interval. Instead, some authors \cite{Judge,Davydov}\footnote{It is relevant here to note that relation (\ref{EqJudge}) is not proved in D.~Judge's paper \cite{Judge}, to which A.\,S.~Davydov refers in 
\cite{Davydov}. Judge proved only the weaker relation $\Delta x\Delta
p\geq0.16\hbar(1-\frac{3}{l^2}\Delta x^2)$ and conjectured that (\ref{EqJudge}) holds.}
propose (along with many other variants, see, for example, \cite{CN,SPb,Trifonov,Dumitru}) the relation 
\begin{equation}\label{EqJudge}
\Delta x\Delta p\geq\frac{\hbar}{2}\left(1-\frac{3}{l^2}\Delta
x^2\right).
\end{equation}
Hence, one can see that  $\Delta x\geq
l/\sqrt3$ for $\Delta p=0$, which will be also shown in the next section. On the other hand, when
$\Delta x\to0$ or $l\to\infty$, we again obtain the usual relation on the line (\ref{EqUncert}).

Second, as mentioned in Subsection~\ref{SecModelDescr}, the energy operator for a particle in the infinite well does not commute with the momentum operator. Usually a quantum particle with small position and momentum dispersions is associated to a classical particle (for example, in the semiclassical limit). But a classical particle is characterized not only by well-defined position and momentum but also by a well-defined energy. Hence, the corresponding quantum particle must have a small energy dispersion as well. Of course, a particle in a nanoscale domain is also expected to have small (or, at least, finite) energy dispersion. For a particle on the line, the energy dispersion is small whenever the momentum dispersion is small, because the momentum and energy operators commute. This is not the case for a particle in the infinite well; therefore, the energy dispersion should be analysed separately.

\subsection{Squeezed quantum states on an interval: definition}\label{SecProblemFormulation}

We are going to construct an analogue of squeezed states for a quantum particle on an interval. To a classical particle with arbitrary position
$x^*\in(-l,l)$ and momentum $p^*\in\mathbb R$, we should assign a quantum wave packet $\psi_{x^*p^*}$ (henceforth we will omit the indices $x^*$ and $p^*$, assuming that they are arbitrary but fixed) for which the following second moments of the position and momentum distributions are small:

\begin{align}\label{EqStanddevX}
\Delta_* x^2&=\int_{-l}^l(x-x^*)^2|\psi(x)|^2\,dx
=\|\hat x\psi-x^*\psi\|^2,\\
\label{EqStanddevP} \Delta_*
p^2&=\sum_{k=-\infty}^{+\infty}
(p_k-p^*)^2|a_k|^2.
\end{align}
Here $a_k$, $k=0,\pm1,\pm2,\ldots$ are the coefficients of the expansion of $\psi$ in the momentum eigenfunctions:

\begin{equation}\label{EqPsiImpSer}
\psi(x)=\frac{1}{\sqrt{2l}}\sum_{k=-\infty}^{+\infty}a_k\,e^{i\frac{\pi}{l}kx}=
\frac{1}{\sqrt{2l}}
\sum_{k=-\infty}^{+\infty}a_k\,e^{i\frac{p_k}{\hbar}x}.
\end{equation}
The following two normalization conditions are equivalent:
$$\int_{-l}^l|\psi(x)|^2\,dx=1,\qquad\sum_{k=-\infty}^{+\infty}|a_k|^2=1.$$

If the vector $\psi$ belongs to the domain of the momentum operator $\hat p$, then formula (\ref{EqStanddevP}) can be compactly rewritten as $\Delta_* p=\|\hat p\psi-\overline
p\psi\|$ or as
\begin{equation}\label{EqPDeltaDeriv}
\Delta_*
p^2=\int_{-l}^l\overline\psi(x)(-i\hbar\frac d{dx}-p^*)^2\psi(x)\,dx.
\end{equation}
However, the original formula (\ref{EqStanddevP}) is more general because it does not assume that the wave packet belongs to the domain of the operator, but only requires that the dispersion should be finite (which is a necessary condition for assigning a quantum wave packet to a classical particle).

The mean values of the position and momentum are calculated by the formulae:
\begin{equation}\label{EqMeanX}\overline
x=\int_{-l}^lx|\psi(x)|^2\,dx,\qquad\overline
p=\sum_{k=-\infty}^{+\infty}p_k|a_k|^2,\end{equation}
Again, if $\psi\in D(\hat p)$, then
\begin{equation}\label{EqPMeanDeriv}
\overline p=\int_{-l}^l\overline\psi(x)(-i\hbar\frac d{dx})\psi(x)\,dx.
\end{equation}

We do not require the exact equalities $\overline x=x^*$ and $\overline p=p^*$; we only require that the quantities $|\overline x-x^*|$ and $|\overline p-p^*|$ should be small. Accordingly, the above moments  $\Delta_* x^2$ and $\Delta_* p^2$ are not, generally speaking, the position and momentum dispersions. However, the smallness of precisely these moments defined by formulae  (\ref{EqStanddevX}) and (\ref{EqStanddevP}) was considered in \cite{Neumann} as a condition for legitimately assigning a quantum wave packet $\psi$ to a classical particle with position $x^*$ and momentum $p^*$.

Denote the standard deviations of the position and momentum by $\Delta x^2$ and $\Delta p^2$ respectively:

\begin{equation}\label{EqStddevX}
\Delta x^2=\int_{-l}^l(x-\overline x)^2|\psi(x)|^2\,dx,\qquad
%\label{EqStanddevP}
\Delta p^2=\sum_{k=-\infty}^{+\infty}
(p_k-\overline p)^2|a_k|^2.
\end{equation}
One can easily show that
\begin{equation}\label{EqDeltaStarStd}
\Delta_* x^2=\Delta x^2+(\overline x-x^*)^2,\qquad
\Delta_* p^2=\Delta p^2+(\overline p-p^*)^2.
\end{equation}
Thus, the smallness of the moments $\Delta_* x$ and $\Delta_* p$ implies both the smallness of $|\overline x-x^*|$ and $|\overline p-p^*|$ and the smallness of the position and momentum dispersions. Conversely, if the mean values of the position and momentum are close to the values of  $x^*$ and $p^*$, and the position and momentum dispersions are small, then $\Delta_* x$ and $\Delta_* p$ are also small.

Note that, for the case of a circle, first formulae (for position) of (\ref{EqStanddevX}), (\ref{EqMeanX}), and (\ref{EqStddevX}) have a meaning only if the considered wave functions are negligible near the point $x=\pm l$ (remind that the points $\pm l$ are identified on the circle). Otherwise it is possible that, for example, a particle is localized near the point $x^*=\pm l$ with nearly zero dispersion, but formula (\ref{EqMeanX}) gives $\overline x=0$ and formulae  (\ref{EqStanddevX}) and (\ref{EqStddevX}) gives a significant dispersion of position. But this is not our case, since we consider wave functions localized near the point $x^*=0$ (see Remark~\ref{RemPosCircle}). In case of the infinite well, $x^*\in(l,l)$ is arbitrary.

As we have already mentioned above, since the momentum spectrum is discrete, there exist wave functions with definite momentum and a finite position dispersion. For example, for the function $\psi(x)=\frac{1}{\sqrt{2l}}e^{i\frac{\pi}{l}kx}$, $k\in\mathbb Z$, the standard deviations in position and momentum are $\Delta x=l/\sqrt3$ and $\Delta p=0$, respectively. This implies (\ref{EqUncertZero}).

We will construct a family $\psi_{x^*,p^*,\alpha}$, $x^*\in[-l,l]$, $p^*\in\mathbb R$, $\alpha>0$ (we will omit indices $x^*$ and $p^*$ in the subsequent), of functions in
$L_2(-l,l)$ that possess the following properties:

\begin{enumerate}[1)]

\item The position and momentum standard deviations $\Delta x_\alpha$ and $\Delta p_\alpha$ defined by (\ref{EqStddevX}) for the state $\psi_\alpha$ are finite for all $\alpha$ and satisfy the relation 
\begin{equation}\label{EqUncertMinAsimpt}
\Delta x_\alpha\Delta p_\alpha\to\frac{\hbar}{2},
\end{equation}
as $\alpha\to\infty$; i.e., the minimal uncertainty relation for coherent states on the line holds asymptotically.

\item $\Delta x_\alpha\to0$ as $\alpha\to\infty$   (i.e., the position dispersion can be made arbitrarily small at the cost of increasing the momentum dispersion).

\item $\overline x\to x^*$ and $\overline p\to p^*$ as $\alpha\to\infty$ (i.e., the mean values of the position and momentum tend to prescribed values).

\item As $l\to\infty$, the functions $\psi_\alpha$ tend to the squeezed states on the real line (\ref{EqCoher}).
\end{enumerate}

The states in such a family will be called \textit{squeezed states on an interval}. Let us comment property (\ref{EqUncertMinAsimpt}). According to one of the definitions,  squeezed states are states for which the uncertainty principle is saturated. However, as we said above, there is no consensus about the ``right'' form of the uncertainty principle for bounded domains. By this reason, we demand squeezed states on an interval to saturate the usual uncertainty relation on the line asymptotically.

We can apply such states to the theory of nanoscale systems
\cite{nano}. The spatial dimensions of nanoscale systems usually range from 1~nm to 100~nm. We will see that there exist values $\alpha$ such that $\Delta x_\alpha$ and
$\Delta p_\alpha$ are simultaneously small while the interval length $l$ is equal to, say, 100~nm. In particular, there exist squeezed states such that $\Delta x\sim0.1$~nm and $\Delta
p\sim10^{-24}$~kg$\cdot$m/s (which corresponds to a minimum energy on the order of $10^{-2}$~eV for the hydrogen atom mass $m\sim10^{-27}$~kg). Such a packet is well localized in the sense that the position dispersion of a particle is less than 0.1~nm and the energy needed for forming such a state is rather small.

\section{Families of squeezed states on an interval}\label{SecSq}

\subsection{Squeezed states given by truncated Gaussian functions}\label{SecPsiGauss}

As a candidate for the considered family of wave functions, we can consider``truncated Gaussian functions'' of the form
$$\psi_{0\beta}(x)=\frac{1}{\sqrt[4]{2\pi\beta^2}}
e^{-\frac{(x-x^*)^2}{4\beta^2}+i\frac{xp^*}{\hbar}}\chi_{[-l,l]}(x),\quad
\beta>0,$$ where $\chi_{[-l,l]}(x)$ is the characteristic function of the interval $[-l,l]$ and $x^*\in(-l,l)$ and $p^*\in\mathbb R$ are given position and momentum of the particle. However, this does not seem to be the best choice, because we have $\psi'_{0\beta}(-l)\neq\psi'_{0\beta}(l)$ for the functions of this family. If 
$x^*\neq0$ or $p^*\neq0$, then also $\psi_{0\beta}(-l)\neq\psi_{0\beta}(l)$. At the same time, if the momentum coefficients $a_k$ decrease rapidly (for example, if  $a_k=O(k^{-2-\varepsilon})$, $\varepsilon>0$), then series (\ref{EqPsiImpSer}) converges uniformly together with the series of derivatives
$$\psi'(x)=\frac{1}{\sqrt{2l}}\sum_{k=-\infty}^{+\infty}i\frac{\pi}{l}k\,a_k\,e^{i\frac{\pi}{l}kx}.$$
Then it is obvious that $\psi(-l)=\psi(l)$ and $\psi'(-l)=\psi'(l)$, because the general term of the series for $\psi(x)$ possesses these properties.

Notice that since the functions $\psi_{0\beta}(x)$ do not belong to the domain of the operator $\hat p^2$ and, when
$x^*\neq0$ or $p^*\neq0$, do not belong even to the domain of $\hat p$, one cannot employ formulae like  (\ref{EqPDeltaDeriv}) or (\ref{EqPMeanDeriv})
in order to calculate the mean value and dispersion of the momentum. 

Moreover, as we will see in Subsection~\ref{SecDisp} the  energy dispersions of such states cannot be finite. In case $x^*\neq0$ or $p^*\neq0$, the momentum dispersion is also infinite.

Consider another family of ``truncated Gaussian functions'', which vanish smoothly in small neighbourhoods of the endpoints of the interval:
\begin{equation}\label{EqGaussCut}
 \psi_\beta(x)=\frac{B_\beta}{\sqrt[4]{2\pi\beta^2}}
e^{-\frac{(x-x^*)^2}{4\beta^2}+i\frac{xp^*}{\hbar}}\eta_\epsilon(x),\quad
\beta>0,\end{equation}
 where
$$\eta_\epsilon(x)=\int_{-\infty}^{+\infty}\chi_\epsilon(y)\,\omega_\epsilon(x-y)\,dy,$$
$$\omega_\epsilon(x)=\begin{cases}C_\epsilon\,
e^{-\frac{\epsilon}{\epsilon-|x|}},&|x|<\epsilon,\\
0,&|x|\geq\epsilon.\end{cases}$$ is a bump function (the constant $C_\epsilon$ is chosen so that
$\int_{-\infty}^{+\infty}\omega_\epsilon(x)\,dx=1$), and
$\chi_\epsilon(x)$  is the
characteristic function of the interval
$[-l+2\epsilon,l-2\epsilon]$. It is obvious that
$0\leq\eta_\epsilon(x)\leq1$ and
$$\eta(x)=\begin{cases}1,&y\in[-l+3\epsilon,l-3\epsilon],\\
0,&y\notin[-l+\epsilon,l-\epsilon].\end{cases}$$ 
In (\ref{EqGaussCut}), $B_\beta$ is a real normalization constant and $\epsilon>0$ is chosen arbitrarily. We will always choose $\epsilon$ so small that at least
$|x^*|<l-3\epsilon$, i.e., $\eta_\epsilon=1$ in a neighbourhood of the point $x^*$.

Unlike $\psi_{0\beta}(x)$, the functions $\psi_\beta(x)$
belong to the domain of the operator $\hat p^m$ for every
$m=1,2,\ldots$.

\begin{lemma}
If $\int_{-l}^l|\psi_\beta(x)|^2\,dx=1$, then
\begin{equation}\label{EqPsiGaussA}
B_\beta=1+O\left(e^{-\frac{(l-|x^*|-3\epsilon)^2}{2\beta^2}}\right),\quad\beta\to0.
\end{equation}
\end{lemma}
\begin{proof}
$$1=\int_{-l}^l|\psi_\beta(x)|^2\,dx\geq\frac{B_\beta^2}{\sqrt{2\pi\beta^2}}
\int_{-l+3\epsilon}^{l-3\epsilon}e^{-\frac{(x-x^*)^2}{2\beta^2}}dx,$$ hence
$$B_\beta\leq\frac{1}{\sqrt{1+O(\beta e^{-\frac{(l-|x^*|-3\epsilon)^2}{2\beta^2}})}}=
1+O\left(\beta e^{-\frac{(l-|x^*|-3\epsilon)^2}{2\beta^2}}\right),$$
where we used the asymptotic formula
(\ref{EqErfcAsimptMain}) for the Gaussian integral (see Appendix~A
below).
\end{proof}

\begin{theorem}\label{TheoPsiGauss}
The following asymptotic formulae are valid for the wave functions $\psi_\beta(x)$, $\beta>0$, defined by (\ref{EqGaussCut}) as $\beta\to0$:

\begin{align}\label{EqPsiGaussMeanX}
&\overline
x_\beta=x^*+O\left(\beta e^{-\frac{(l-|x^*|-3\epsilon)^2}{2\beta^2}}\right),\\
\label{EqPsiGaussMeanP} &\overline
p_\beta=p^*,\\\label{EqPsiGaussStdDevX}
&\Delta_*x_\beta^2=\beta^2+O\left(\beta e^{-\frac{(l-|x^*|-3\epsilon)^2}{2\beta^2}}\right),\\
\label{EqPsiGaussStdDevP}
&\Delta_*p_\beta^2=\left(\frac{\hbar}{4\beta}\right)^2+O\left(\beta^{-3}
e^{-\frac{(l-|x^*|-3\epsilon)^2}{2\beta^2}}\right).
\end{align}
\end{theorem}
\begin{proof}
\begin{multline*}
\overline x_\beta=\int_{-l}^lx|\psi_\beta(x)|^2\,dx=
x^*+\int_{-l-x^*}^{l-x^*}x|\psi_\beta(x+x^*)|^2\,dx\\=
x^*+\frac{B_\beta^2}{\sqrt{2\pi\beta^2}}
\int_{-l-x^*+3\epsilon}^{l-x^*-3\epsilon}xe^{-\frac{x^2}{2\beta^2}}\,dx+
\frac{B_\beta^2}{\sqrt{2\pi\beta^2}}
\left(\int_{-l-x^*}^{-l-x^*+3\epsilon}+
\int_{l-x^*-3\epsilon}^{l-x^*}\right)
xe^{-\frac{x^2}{2\beta^2}}\eta_\epsilon(x)\,dx\\=
x^*+O\left(\beta
e^{-\frac{(l-|x^*|-3\epsilon)^2}{2\beta^2}}\right),
\end{multline*}
where we used the asymptotic formula
(\ref{EqErfcAsimptMain}) for the Gaussian integral, formula
(\ref{EqPsiGaussA}), and the estimate $|\eta_\epsilon(x)|\leq1$.
 Formula (\ref{EqPsiGaussMeanX})
is proved.
\begin{multline*}
\overline p_\beta=\int_{-l}^l\overline{\psi_\beta(x)}(-i\hbar)\psi'_\beta(x)dx\\=
\int_{-l}^l|\psi_\beta(x)|^2(p^*+\frac{i\hbar(x-x^*)}{2\beta^2})dx-
\frac{i\hbar B_\beta^2}{\sqrt{2\pi\beta^2}}
\int_{-l}^le^{-\frac{(x-x^*)^2}{2\beta^2}}\eta_\epsilon(x)\eta'_\epsilon(x)dx.
\end{multline*}
The imaginary terms must cancel each other out because the diagonal matrix elements of a self-adjoint operator must be real. Therefore,
$\overline p_\beta=p^*$; i.e., we have obtained formula
(\ref{EqPsiGaussMeanP}).

$$\Delta_*x_\beta^2=\int_{-l}^l(x-x^*)^2|\psi_\beta(x)|^2\,dx=
\int_{-l-x^*}^{l-x^*}x^2|\psi_\beta(x+x^*)|^2\,dx= \beta^2+
O\left(\beta e^{-\frac{(l-|x^*|-3\epsilon)^2}{2\beta^2}}\right).$$
Here we applied the asymptotic formula
(\ref{EqErfcDerivAsimptMain}) from Appendix~\ref{SecAsympGauss}. Formula (\ref{EqPsiGaussStdDevX})
is proved.

\begin{multline*}
\Delta_*p_\beta^2=-\hbar^2\int_{-l}^l\overline{\psi_\beta(x)}\psi''_\beta(x)dx-(p^*)^2\\=
\int_{-l}^l|\psi_\beta(x)|^2\left[(p^*)^2-\frac{\hbar^2(x-x^*)^2}{4\beta^4}+
\frac{\hbar^2}{2\beta^2}-\frac{ip^*\hbar(x-x^*)}{\beta^2}\right]dx\\+
\frac{B_\beta^2}{\sqrt{2\pi\beta^2}}
\int_{-l}^le^{-\frac{(x-x^*)^2}{2\beta^2}}\left[
-\hbar^2\eta_\epsilon(x)\eta''_\epsilon(x)
+2\eta_\epsilon(x)\eta'_\epsilon(x)\left(\frac{\hbar(x-x^*)}{2\beta^2}-i\hbar
p^*\right)\right]dx-(p^*)^2.\end{multline*} Again, the imaginary terms must cancel each other out because the diagonal matrix elements of a self-adjoint operator must be real. We have
\begin{multline*}\Delta_*p_\beta^2=
\int_{-l}^l|\psi_\beta(x)|^2\left(\frac{\hbar^2}{2\beta^2}-\frac{\hbar^2(x-x^*)^2}{4\beta^4}\right)dx+
\\+\frac{B_\beta^2}{\sqrt{2\pi\beta^2}}
\int_{-l}^le^{-\frac{(x-x^*)^2}{2\beta^2}}\left[\eta_\epsilon(x)\eta'_\epsilon(x)
\frac{\hbar(x-x^*)}{\beta^2}
-\hbar^2\eta_\epsilon(x)\eta''_\epsilon(x)\right]
dx.\end{multline*} Here
$$\int_{-l}^l|\psi_\beta(x)|^2\frac{\hbar^2(x-x^*)^2}{4\beta^4}dx=
\left(\frac{\hbar}{2\beta}\right)^2+
O\left(\beta^{-3}e^{-\frac{(l-|x^*|-3\epsilon)^2}{2\beta^2}}\right).$$
Since the functions $\eta_\epsilon(x)$, $\eta'_\epsilon(x)$, and
$\eta''_\epsilon(x)$ are bounded (for a fixed $\epsilon$) and since $\eta'_\epsilon(x)$ and
$\eta''_\epsilon(x)$ are different from zero only on the intervals
$[-l+\epsilon,-l+3\epsilon]$ and
$[l-3\epsilon,l-\epsilon]$, we obtain the following estimate:
$$\frac{B_\beta^2}{\sqrt{2\pi\beta^2}}
\int_{-l}^le^{-\frac{(x-x^*)^2}{2\beta^2}}\left[\eta_\epsilon(x)\eta'_\epsilon(x)
\frac{\hbar(x-x^*)}{\beta^2}
-\hbar^2\eta_\epsilon(x)\eta''_\epsilon(x)\right] dx=
O\left(\beta^{-1}e^{-\frac{(l-|x^*|-3\epsilon)^2}{2\beta^2}}\right).$$
Therefore,
$$\Delta_*p_\beta^2=\left(\frac{\hbar}{2\beta}\right)^2+
O\left(\beta^{-3}e^{-\frac{(l-|x^*|-3\epsilon)^2}{2\beta^2}}\right).$$ This proves formula
(\ref{EqPsiGaussStdDevP}) and completes the proof of the theorem.\end{proof}

\begin{corollary}
The following asymptotic relation holds for the wave functions $\psi_\beta(x)$, $\beta>0$, defined by formula (\ref{EqGaussCut}) as $\beta\to0$
\begin{equation}\label{EqPsiGaussUncert}
\Delta x_\beta^2\Delta p_\beta^2=\frac{\hbar^2}{4}+O\left(\beta^{-1}
e^{-\frac{(l-|x^*|-3\epsilon)^2}{2\beta^2}}\right)\end{equation}
(i.e., relation (\ref{EqUncertMinAsimpt}) holds).
\end{corollary}

Thus, we have $\Delta x_\beta\to0$ and $\Delta p_\beta\to\infty$ as $\beta\to0$. It is obvious that $\Delta p_\beta\to0$ and $\Delta x_\alpha\to l/\sqrt 3$ as $\beta\to\infty$ and $\epsilon\to0$  (because $\psi(x)\to1/\sqrt{2l}$ in $L_2(-l,l)$).

For sufficiently small $\beta$ (such that one can apply the asymptotic estimates from Theorem~\ref{TheoPsiGauss}), $\Delta x$ and $\Delta p$, are estimated by quantities of the same order as for coherent states on the line: $\Delta x\sim0.1$~nm and
$\Delta p\sim10^{-24}$~kg$\cdot$m/s.

\subsection{Squeezed states in the form of the theta function}\label{SecPsiTheta}

Here we present another method for constructing a family of wave functions with required properties. Given a position $x^*\in(-l,l)$ and a momentum $p^*\in\mathbb R$, define a family of functions in $L_2(-l,l)$ for $\alpha>0$ as follows:

\begin{equation*}%\label{EqPsiTheta0}
\psi_\alpha(x)=\frac{1}{\sqrt{2l}}\sum_{k=-\infty}^{+\infty}a_k^{(\alpha)}
e^{i\frac{\pi}{l}k(x-x^*)},\end{equation*}
where
$$a_k^{(\alpha)}=A_\alpha e^{-\frac{(k-k^*)^2}{4\alpha^2}},$$
$k^*$ is the nearest integer to $\frac{l}{\pi}\frac{p^*}{\hbar}$, and $A_\alpha$ is a real normalization constant.

Such a wave function can be represented by the theta function as:
\begin{equation}\label{EqPsiTheta}
\psi_\alpha(x)=\frac{A_\alpha}{\sqrt{2l}}\,\theta\left(\frac{x-x^*}{2l},
\frac{1}{4\pi\alpha^2}\right)e^{i\frac{\pi}{l}k^*(x-x^*)}\end{equation}
(see formula (\ref{EqTheta}) and Appendix~\ref{SecTheta}). These quantum states generalize (by means of an arbitrary $\alpha$) the coherent and squeezed states on the circle introduced in \cite{BG,GdOCircle,KRCircle,KR2002,HM}.

\begin{lemma} If $\int_{-l}^l|\psi_\alpha(x)|^2\,dx=1$, then
\begin{equation}\label{EqPsiThetaA}
A_\alpha=\frac{1}{\sqrt[4]{2\pi\alpha^2}}+O\left(\frac{1}{\sqrt\alpha} \,e^{-2(\pi\alpha)^2}\right),\quad \alpha\to\infty.\end{equation}
\end{lemma}
\begin{proof}
We have
$$1=\sum_{k=-\infty}^{+\infty}|a_k^{(\alpha)}|^2=A_\alpha^2\sum_{k=-\infty}^{+\infty}
e^{-\frac{(k-k^*)^2}{2\alpha^2}}=A_\alpha^2\,\theta\left(0,\frac{1}{2\pi\alpha^2}\right)=
A_\alpha^2\,[\sqrt{2\pi\alpha^2}+O(\alpha\, e^{-2(\pi\alpha)^2})]$$
as $\alpha\to\infty$. Here we used formula (\ref{EqThetaVal}). Hence,
$$A_\alpha=\frac{1}{\sqrt{\sqrt{2\pi\alpha^2}+O(\alpha e^{-2(\pi\alpha^2)})}}=
\frac{1}{\sqrt[4]{2\pi\alpha^2}}+O\left(\frac{1}{\sqrt\alpha} \,e^{-2(\pi\alpha)^2}\right),\quad\alpha\to\infty.$$
\end{proof}

\begin{theorem}\label{TheoPsiTheta}
The following asymptotic estimates hold for the wave functions $\psi_\alpha(x)$, $\alpha>0$, defined by formula (\ref{EqPsiTheta}) as $\alpha\to\infty$:

\begin{align}\label{EqPsiThetaVal}
&\psi_\alpha(x)=\sqrt[4]{\frac{2\pi\alpha^2}{l^2}}
e^{-\left(\alpha\pi d(\frac{x-x^*}{l})\right)^2+i\frac{\pi}{l}k^*(x-x^*)}+O(\sqrt\alpha e^{-(\pi\alpha)^2}),\\\label{EqPsiThetaMeanX}
&\overline x_\alpha-x^*= l\,O\left(\alpha^{-1}e^{-2\left[\pi\alpha\left(1-\frac{|x^*|}{l}\right)\right]^2}\right),\\
\label{EqPsiThetaMeanP}
&|\overline p_\alpha-p^*|\leq\frac{\pi}{l}\hbar,\\\label{EqPsiThetaStdDevX}
&\Delta_*x_\alpha^2=\left(\frac{l}{2\pi\alpha}\right)^2+
l^2O\left(\alpha^{-1}e^{-2\left[\pi\alpha\left(1-\frac{|x^*|}{l}\right)\right]^2}\right),\\\label{EqPsiThetaStdDevP}
&\Delta_*p_\alpha^2=\left(\frac{\pi}{l}\hbar\alpha\right)^2[1+
O(e^{-2(\pi\alpha)^2})].
\end{align}
Here $0\leq d(x)\leq\frac{1}{2}$ is the distance on the real line from the point $x$ to the nearest integer.
\end{theorem}
\begin{proof}
Substituting formulae (\ref{EqPsiThetaA}) and (\ref{EqThetaVal}) into (\ref{EqPsiTheta}), we obtain
\begin{multline*}
\psi_\alpha(x)=
\frac{1}{\sqrt{2l}}e^{i\frac{\pi}{l}k^*(x-x^*)}\left[\frac{1}{\sqrt[4]{2\pi\alpha^2}}+
O\left(\frac{1}{\sqrt\alpha} \,e^{-2(\pi\alpha)^2}\right)\right]\\\times
\left[\sqrt{4\pi\alpha^2}\,e^{-4\pi\alpha d(\frac{x-x^*}{2l})]^2}+
O\left(\alpha e^{-4\left[\pi\alpha\left(1-d(\frac{x-x^*}{2l})\right)\right]^2}\right)\right]\\=
\sqrt[4]{\frac{2\pi\alpha^2}{l^2}}\,
e^{-\left(\alpha\pi d(\frac{x-x^*}{l})\right)^2+i\frac{\pi}{l}k^*(x-x^*)}+O(\sqrt\alpha\, e^{-(\pi\alpha)^2}).\end{multline*}
We get formula (\ref{EqPsiThetaVal}).

\begin{multline*}
\overline x_\alpha=\int_{-l}^{l}x|\psi_\alpha(x)|^2\,dx=
\int_{-l}^{l}x^*|\psi_\alpha(x)|^2\,dx+\int_{-l}^{l}(x-x^*)|\psi_\alpha(x)|^2\,dx
\\=
x^*+\int_{-l-x^*}^{l-x^*}x|\psi_\alpha(x+x^*)|^2\,dx=
x^*+\frac{A_\alpha^2}{2l}\int_{-l-x^*}^{l-x^*}x\left|
\theta\left(\frac{x}{2l},\frac{1}{4\pi\alpha^2}\right)\right|^2\,dx\\=
x^*+2lA_\alpha^2\int_{-\frac{1}{2}-\frac{x^*}{2l}}^
{\frac{1}{2}-\frac{x^*}{2l}}
y\left|\theta\left(y,\frac{1}{4\pi\alpha^2}\right)\right|^2\,dy=
x^*+l\,O\left(\alpha^{-1}e^{-2[\pi\alpha(1-\frac{|x^*|}{l})]^2}\right)
\end{multline*}
(we used formulae (\ref{EqPsiThetaA}) and (\ref{EqThetaMeanX}) in the last equality). Thus, formula
(\ref{EqPsiThetaMeanX}) is proved.

Estimate (\ref{EqPsiThetaMeanP}) follows from the fact that
$$\overline p_\alpha=\frac{\pi}{l}\hbar\sum_{k=-\infty}^{+\infty}
k|a^{(\alpha)}_k|^2=\frac{\pi}{l}\hbar k^*,$$
and the definition of $k^*$ as the nearest integer to $(\frac{l}{\pi}\frac{p^*}{\hbar})$.

\begin{multline*}
\Delta_* x_\alpha^2=\int_{-l}^{l}(x-x^*)^2|\psi_\alpha(x)|^2\,dx=
\int_{-l-x^*}^{l-x^*}x^2|\psi_\alpha(x+x^*)|^2\,dx\\=
\frac{A_\alpha^2}{2l}\int_{-l-x^*}^{l-x^*}x^2\left|
\theta\left(\frac{x}{2l},\frac{1}{4\pi\alpha^2}\right)\right|^2\,dx=
(2l)^2A_\alpha^2\int_{-\frac{1}{2}-\frac{x^*}{2l}}^
{\frac{1}{2}-\frac{x^*}{2l}}
y^2\left|\theta\left(y,\frac{1}{4\pi\alpha^2}\right)\right|^2\,dy\\=
(2l)^2\left[\frac{1}{\sqrt{2\pi}\alpha}+
O\left(\frac{1}{\alpha}e^{-(\pi\alpha)^2}\right)\right]
\left[\frac{1}{4\pi\sqrt{8\pi\alpha^2}}+
O\left(e^{-2[\pi\alpha(1-\frac{|x^*|}{l})]^2}\right)\right]\\=
\left(\frac{l}{2\pi\alpha}\right)^2+
l^2O\left(\alpha^{-1}e^{-2\left[\pi\alpha\left(1-\frac{|x^*|}{l}\right)\right]^2}\right)\end{multline*}
(here we used formulae (\ref{EqPsiThetaA}) and (\ref{EqThetaStdDevX})). Formula
(\ref{EqPsiThetaStdDevX}) is proved.

\begin{multline*}
\Delta_*p_\alpha^2=\left(\frac{\pi}{l}\hbar\right)^2\sum_{k=-\infty}^{+\infty}
(k-k^*)^2|a^{(\alpha)}_k|^2=
A_\alpha^2\left(\frac{\pi}{l}\hbar\right)^2\sum_{k=-\infty}^{+\infty}
k^2e^{-\frac{k^2}{2\alpha^2}}\\=
\left(\frac{\pi}{l}\hbar\right)^2\left[\frac{1}{\sqrt{2\pi}\alpha}+
O\left(\frac{1}{\alpha}e^{-(\pi\alpha)^2}\right)\right]
[\sqrt{2\pi}\alpha^3+O(\alpha^3e^{-2(\pi\alpha)^2})]=
\left(\frac{\pi}{l}\hbar\right)^2[\alpha^2+
O(\alpha^2e^{-2(\pi\alpha)^2})]\end{multline*}
(here we used formulae (\ref{EqPsiThetaA}) and (\ref{EqThetaStdDevK})). Formula
(\ref{EqPsiThetaStdDevP}) is proved. This completes the proof of the theorem.
\end{proof}

\begin{corollary}
The following asymptotic formula holds for the wave functions $\psi_\alpha(x)$, $\alpha>0$, defined by formula (\ref{EqPsiTheta}) as
$\alpha\to\infty$

\begin{equation}\label{EqPsiThetaUncert}\Delta x_\alpha^2\Delta p_\alpha^2=\frac{\hbar^2}{4}+O\left(\alpha
e^{-2\left[\pi\alpha\left(1-\frac{|x^*|}{l}\right)\right]^2}\right)\end{equation}
(i.e., relation (\ref{EqUncertMinAsimpt}) holds).
\end{corollary}

Comparison of formulae (\ref{EqPsiGaussStdDevX}) and (\ref{EqPsiGaussStdDevP}) with (\ref{EqPsiThetaStdDevX}) and (\ref{EqPsiThetaStdDevP}) respectively shows that there is a correspondence between the parameters $\alpha$ and 
$\beta$ given by the relation
$\alpha=\frac{l}{2\pi\beta}$. Therefore, we can write
(\ref{EqPsiThetaUncert}) as
$$\Delta x_\beta^2\Delta p_\beta^2=\frac{\hbar^2}{4}+O\left(\beta^{-1}
e^{-\frac{(l-|x^*|)^2}{2\beta^2}}\right).$$ Comparing with
(\ref{EqPsiGaussUncert}), we see that the left-hand side tends to $\hbar^2/4$ somewhat faster than the truncated Gaussian function, because $\epsilon>0$. At the same time, $\epsilon$ can be arbitrarily small; therefore, the difference between the rates of convergence can be made arbitrarily small.

However, one can notice the faster decrease of the remainder term for $\Delta_*p^2$ in the case of the theta function (the exponential function in the remainder term is multiplied by $\alpha^2$ in (\ref{EqPsiThetaStdDevP}) and by $\beta^{-3}$) in (\ref{EqPsiGaussStdDevP})).

Thus, $\Delta x_\alpha\to0$ and
$\Delta p_\alpha\to\infty$ as $\alpha\to\infty$. It is obvious that $\Delta
p_\alpha\to0$ as $\alpha\to0$, because $a_{k}\to\delta_{k\overline k}$. Then
$\Delta x_\alpha\to l/\sqrt 3$. For sufficiently large $\alpha$ (such that one can apply the asymptotic estimates from Theorem~\ref{TheoPsiTheta}), the estimates $\Delta x\sim0.1$~nm and $\Delta p\sim10^{-24}$~kg$\cdot$m/s again hold.

We have established asymptotic minimization of uncertainty relation (\ref{EqUncertMinAsimpt}). It would be interesting to find states with finite $\Delta x$ and $\Delta p$ that turn the uncertainty relation (\ref{EqJudge}) into an equality. We suppose that this uncertainty relation may be minimized by functions $\psi_\alpha$ from the family constructed here on the basis of the theta function.

\subsection{The case of an arbitrary density function}\label{SecArbitrary}

In the previous subsection, the construction of a family of wave functions with required properties was based on the density of the Gaussian distribution of momentum. Here we describe a general method for constructing such a family, where the distribution of  momentum is rather arbitrary.

Let, again, $x^*\in(-l,l)$ and $p^*\in\mathbb R$ be given position and momentum of a particle. Denote  $$k^*=\frac{l}{\pi}\frac{p^*}{\hbar}.$$ Let
$\varphi(q)$ be a density function on the line with zero mean, i.e., a function such that $\varphi(q)\geq0$ for $q\in\mathbb
R$ and 
$$\int_{-\infty}^{+\infty}\varphi(q)\,dq=1,\quad\int_{-\infty}^{+\infty}q\varphi(q)\,dq=0.$$ We also require that the second moment of $\varphi(q)$ be finite, denote it by
$$\Delta q^2=\int_{-\infty}^{+\infty}q^2\varphi(q)\,dq.$$

Introduce a family of functions $\{\varphi_\alpha(k)\}_{\alpha>0}$ by the formula
\begin{equation}\label{EqPhiAlpha}
\varphi_{\alpha l}(q)=\frac{1}{\alpha}\,\varphi_l\left(\frac{q-\overline k}{\alpha}\right),
\end{equation}
\begin{equation}\label{EqPhiL}
\varphi_l(q)=\frac{\pi}{l}\varphi\left(\frac{\pi}{l}q\right)
\end{equation}
where $\overline k$ is the nearest integer to $k^*$.  Then
\begin{equation}
\label{EqMeanPDev}|\overline p-p^*|\leq\frac{\pi}{l}\hbar.
\end{equation}

The following relations hold:
\begin{equation}\label{EqPhiAlphaRelations}\begin{aligned}
&\int_{-\infty}^{+\infty}\varphi_{\alpha l}(q)\,dq=1,\\
&\int_{-\infty}^{+\infty}q\varphi_{\alpha l}(q)\,dq=\overline k,\\
\Delta
q_\alpha^2\equiv&\int_{-\infty}^{+\infty}(q-\overline
k)^2\varphi_{\alpha l}(q)\,dq= (\frac{\alpha l}\pi)^2\Delta
q^2.\end{aligned}
\end{equation} 
Thus, $\{\varphi_{\alpha l}\}$, $\alpha>0$,
is a family of density functions with the same means and with standard deviations that increase proportionally to $\alpha$.

Set
\begin{equation}\label{EqAk}a_k^{(\alpha)}=
\left[\int_{k-\frac{1}{2}}^{k+\frac{1}{2}}
\varphi_{\alpha l}(q)\,dq \right]^{\frac{1}{2}},\end{equation} where $k=0,\pm1,\pm2,\ldots$;
\begin{equation}\label{EqPsin}
\psi_\alpha(x)=
\frac{1}{\sqrt{2l}}\sum_{k=-\infty}^{+\infty}a^{(\alpha)}_k\,e^{i\frac{\pi}{l}k(x-x^*)},
\end{equation}
where $x^*\in(-l,l)$. This construction is similar to that proposed in \cite{GCircle}.

Denote the mean values and the standard deviations of the position and momentum
for the wave functions $\psi_\alpha$ by $\overline x_\alpha$,
$\overline p_\alpha$, $\Delta x_\alpha$, and $\Delta p_\alpha$.

\begin{theorem}\label{TheoPackOtrExist}
Suppose that the density function $\varphi(q)$ is even, has a maximum at zero and does not increase as  $|q|$ increases (in particular, this means that the local maximum at the point
$q=0$ is also a global maximum).

Then the following inequalities and relations hold:
\begin{equation}\label{EqStanddevXBound}
\Delta_* x_\alpha^2\leq\frac{9\pi\varphi(0)l}{2\alpha}
\int_{-1}^1\frac{(y-\frac{x^*}{l})^2}
{\sin^2\left(\frac{\pi}{2}(y-\frac{x^*}{l})\right)}\,dy,\end{equation}

\begin{equation}\label{EqStanddevPBound}
\lim_{\alpha\to\infty}\frac{\Delta_* p_\alpha}{\alpha}=C\neq0,
\end{equation}

\begin{equation}\label{EqMeanXDev}
|\overline x_\alpha-x^*|\leq\frac{x^*}{\alpha l}\frac{18\pi\varphi(0)}{\cos^2\frac{\pi x^*}{2l}},\end{equation}
\begin{equation}\label{EqMeanPEq}\overline p_\alpha=\frac{\pi}{l}\hbar \overline k\equiv\overline p.\end{equation}
\end{theorem}

To prove this theorem, we will need two lemmas.
\begin{lemma}\label{LemDeltaP}
For an arbitrary function $\varphi(k)$ satisfying the conditions of Theorem~\ref{TheoPackOtrExist},

1) the mean momentum satisfies the relation
$$\overline p_\alpha=\frac{\pi}{l}\hbar \overline k\equiv \overline p;$$

2a) the momentum standard deviation can be estimated as

$$-\frac{1}{12}\left(\frac{\pi}{l}\hbar\right)^2
\left[1+\frac{2}{\alpha}\varphi_l(0)\right]
\leq\Delta_* p_\alpha^2-\widetilde\Delta p_\alpha^2\leq\frac{1}{6}\left(\frac{\pi}{l}\hbar\right)^2
\left[1+\frac{2}{\alpha}\varphi_l(0)\right],$$
where $\widetilde\Delta p_\alpha=(\frac{\pi}{l}\hbar\Delta q_\alpha)^2=(\alpha\hbar\Delta q)^2$.

2b) if the function $\varphi(k)$ is twice continuously differentiable,
$\varphi''(k)=O(1/k^2)$ as $k\to\pm\infty$, and $\varphi''(k)$ has a finite number of local extrema, then the following sharper
result is valid:
$$\Delta_* p_\alpha^2=\widetilde\Delta p_\alpha^2+
\left(\frac{\pi}{l}\hbar\right)^2\left[\frac{1}{12}+O\left(\frac{1}{\alpha^2}\right)
\right],\quad\alpha\to\infty.$$

\end{lemma}

The proof of the lemma is given in Appendix~\ref{SecLemDeltaP}.

\begin{lemma}\label{LemSumCos}
Let $\{a_k\}_{k=0}^{\infty}$ be a nonzero monotonic square-summable (i.e., $\sum_{k=0}^{\infty}a_k^2<\infty$)
sequence of real numbers. Then the function
\begin{equation}\label{EqLemChi}
\chi(x)=\sum_{k=-\infty}^{+\infty}a_k\cos{kx}
\end{equation}
satisfies the estimate

\begin{equation}\label{EqLemChiBound}
|\chi(x)|\leq\frac{|a_0|}{|\sin\frac{x}{2}|}\end{equation}
for $x\neq2\pi n$, $n=0,\pm1,\pm2,\ldots$.
\end{lemma}

The proof of the lemma is given in Appendix~\ref{SecLemSumCos}.

\begin{proof}[Proof of Theorem~\ref{TheoPackOtrExist}]
First of all, notice that formula (\ref{EqAk}) and evenness of $\varphi(q)$ imply
\begin{equation}\label{EqTheoAkSym}
a^{(\alpha)}_{\overline k+k}=a^{(\alpha)}_{\overline k-k}
\end{equation}
for all $k$ and $\alpha$.

Let us prove estimate (\ref{EqStanddevXBound}). In view of (\ref{EqTheoAkSym}), we have
\begin{multline}\label{EqTheoPsiCos}\psi_\alpha(x)=
\frac{1}{\sqrt{2l}}\sum_{k=-\infty}^{+\infty}a^{(\alpha)}_k
\,e^{i\frac{\pi}{l}k(x-x^*)}=
\frac{1}{\sqrt{2l}}\sum_{k=-\infty}^{+\infty}a^{(\alpha)}_{\overline k+k}
\,e^{i\frac{\pi}{l}(\overline k+k)(x-x^*)}\\=
\frac{1}{\sqrt{2l}}\left[2\sum_{k=0}^{+\infty}a^{(\alpha)}_{\overline k+k}
\cos\left(\frac{\pi}{l}k(x-x^*)\right)-
a^{(\alpha)}_{\overline k}\right]e^{i\frac{\pi}{l}\overline k(x-x^*)}.\end{multline}
Then it follows form Lemma~\ref{LemSumCos} that
\begin{multline}\label{EqTheoPsiBound}
|\psi_\alpha(x)|\leq\frac{1}{\sqrt{2l}}
\frac{3|a^{(\alpha)}_{\overline k}|}{|\sin\frac{\pi(x-x^*)}{2l}|}=
\frac{3}{\sqrt{2l}}
\left[\frac{1}{\alpha}
\int_{\overline k-\frac{1}{2}}^{\overline k+\frac{1}{2}}
\varphi_l\left(\frac{q-\overline k}{\alpha}\right)\,dq\right]^{\frac{1}{2}}\frac{1}{|\sin\frac{\pi(x-x^*)}{2l}|}\leq\\\leq
3\sqrt{\frac{1}{2\alpha l}\varphi_l(0)}\,
\frac{1}{|\sin\frac{\pi(x-x^*)}{2l}|}=
3\sqrt{\frac{\pi}{2\alpha l^2}\varphi(0)}\,
\frac{1}{|\sin\frac{\pi(x-x^*)}{2l}|}\end{multline}
for any $x\in[-l,l]\backslash\{x^*\}$. Hence, by formula (\ref{EqStanddevX}),
$$\Delta_* x_\alpha^2\leq
\frac{9\pi\varphi(0)}{2\alpha l^2}
\int_{-l}^l\frac{(x-x^*)^2}{\sin^2\frac{\pi(x-x^*)}{2l}}\,dx=
\frac{9\pi\varphi(0)l}{2\alpha}
\int_{-1}^1\frac{(y-\frac{x^*}{l})^2}
{\sin^2\left(\frac{\pi}{2}(y-\frac{x^*}{l})\right)}\,dy.$$

Formula (\ref{EqStanddevPBound}) follows immediately from Lemma~\ref{LemDeltaP} and the third relation in (\ref{EqPhiAlphaRelations}).

Let us prove inequality (\ref{EqMeanXDev})
$$\overline x=\int_{-l}^lx|\psi(x)|^2\,dx=\int_{-l}^l(x-x^*)|\psi(x)|^2\,dx+x^*=
\int_{-l-x^*}^{l-x^*}x|\psi(x+x^*)|^2\,dx+x^*.$$
Let $x^*\geq0$. By virtue of (\ref{EqTheoPsiCos}) $|\psi(x^*+x)|=|\psi(x^*-x)|$ for any $x\in[l-x^*,l+x^*]$. Therefore,
$$\overline x-x^*=\int_{-l-x^*}^{-l+x^*}x|\psi(x+x^*)|^2\,dx.$$
On the one hand, the integral on the right-hand side is nonpositive because $x\leq0$ on the integration interval. On the other hand, using (\ref{EqTheoPsiBound}), we get
\begin{multline*}
\int_{-l-x^*}^{-l+x^*}x|\psi(x+x^*)|^2\,dx\geq
\frac{9\pi\varphi(0)}{2\alpha l^2}
\int_{-l-x^*}^{-l+x^*}\frac{x\,dx}{\sin^2\frac{\pi x}{2l}}\geq
\frac{9\pi\varphi(0)}{2\alpha l^2}\frac{1}{\sin^2\frac{\pi(l+x^*)}{2l}}
\int_{-l-x^*}^{-l+x^*}x\,dx\\\geq
-\frac{x^*}{\alpha l}\frac{18\pi\varphi(0)}{\cos^2\frac{\pi x^*}{2l}}.
\end{multline*}
Thus,
$$-\frac{x^*}{\alpha l}\frac{18\pi\varphi(0)}{\cos^2\frac{\pi x^*}{2l}}\leq \overline x-x^*\leq0.$$
Similarily, if $x^*\leq0$, we obtain
$$0\leq\overline x-x^*\leq\frac{x^*}{\alpha l}\frac{18\pi\varphi(0)}{\cos^2\frac{\pi x^*}{2l}}.$$
Estimate (\ref{EqMeanXDev}) is proved. Equality (\ref{EqMeanPEq}) is proved as assertion 1) of Lemma~\ref{LemDeltaP}.

This completes the proof of the theorem.
\end{proof}

\begin{corollary}
The following estimate holds for the wave functions $\psi_\alpha(x)$, $\alpha>0$, defined by formula (\ref{EqPsin}):
\begin{equation}\label{EqStdDevXP}\Delta x_\alpha^2\Delta p_\alpha\leq\frac{9}{2}\pi l\hbar\varphi(0)\Delta q\int_{-1}^1\frac{(y-\frac{x^*}{l})^2}
{\sin^2\left(\frac{\pi}{2}(y-\frac{x^*}{l})\right)}\,dy\,
\sqrt{1+\left(\frac{\pi}{l\alpha\Delta q}\right)^2\delta}.
\end{equation} 
Here
$\delta=\frac{1}{6}+\frac{1}{3\alpha}\varphi_l(0)$. If the conditions of assertion 2b) of Lemma~\ref{LemDeltaP} hold, then
$\delta=\frac{1}{12}+O(\alpha^{-2})$ as $\alpha\to\infty$.
\end{corollary}

\begin{proof}
According to Lemma~\ref{LemDeltaP}
$$\Delta_* p_\alpha^2=\widetilde\Delta p_\alpha^2+\left(\frac{\pi}{l}\hbar\right)^2\delta.$$
In view of (\ref{EqPhiAlphaRelations}), we have
$$\widetilde\Delta p_\alpha^2=\left(\frac{\pi}{l}\hbar\Delta q_\alpha\right)^2=
\left(\hbar\alpha\Delta q\right)^2,$$ hence
\begin{equation*}%\label{EqDeltaPFull}
\Delta_* p_\alpha^2=
\left(\hbar\alpha\Delta q\right)^2+\left(\frac{\pi}{l}\hbar\right)^2\delta,\end{equation*}
Since $\Delta x_\alpha\leq\Delta_*x_\alpha$ and
$\Delta p_\alpha\leq\Delta_*p_\alpha$ due to (\ref{EqDeltaStarStd}), the required estimate (\ref{EqStdDevXP}) follows from the above relation and (\ref{EqStanddevXBound}).
\end{proof}

As $\alpha\to\infty$ (and since $\hbar$ is small), we can neglect the last factor (square root expression) on the right-hand side of (\ref{EqStdDevXP}), then we obtain
\begin{equation}\label{EqStdDevXPAppr}\Delta x_\alpha^2\Delta p_\alpha\lesssim\frac{9}{2}\pi l\hbar\varphi(0)\Delta
q\int_{-1}^1\frac{(y-\frac{x^*}{l})^2}
{\sin^2\left(\frac{\pi}{2}(y-\frac{x^*}{l})\right)}\,dy.\end{equation}

Let
$\varphi(k)=\frac{1}{\sqrt{2\pi}}e^{-\frac{k^2}{2}}$ (then
$\varphi(\overline k)=1/\sqrt{2\pi}$, $\Delta q=1$),
$l=100$~nm, $x^*=0$. Then, taking into account that
$\hbar\approx1.05\cdot10^{-34}$ J$\cdot$s and
$$\int_{-1}^1\frac{y^2}
{\sin^2\frac{\pi y}{2}}\,dy\approx1.12,$$ we see that formula (\ref{EqStdDevXPAppr})  can guarantee the condition $\Delta x_\alpha\lesssim0.1$~nm only when $\Delta
p_\alpha\sim10^{-20}$~kg$\cdot$m/s. Comparing this with analogous results obtained in the two previous subsections, where we used specific techniques for the Gaussian integral and the theta function, we see that the estimate obtained for $\Delta x_\alpha$
is rather rough. But it is still enough for nanoscale systems. Moreover, Theorem~\ref{TheoPackOtrExist} gives estimates for finite
$\alpha$, rather than only asymptotic estimates as $\alpha\to\infty$.

As above, we have $\Delta x_\alpha\to0$ and $\Delta p_\alpha\to\infty$ as $\alpha\to\infty$, whereas   $\Delta x_\alpha\to l/\sqrt 3$ and $\Delta p_\alpha\to0$ as $\alpha\to0$.

\section{Further problems}\label{SecFurther}

\subsection{Energy dispersion}\label{SecEnergy}

As we know, a classical particle is characterized not only by well-defined position $x^*$ and momentum $p^*$, but also a well-defined energy $E^*$. For a free particle, we have
$E^*=p^{*2}/2m.$ Hence, to associate a classical particle with a quantum wave packet, the latter must have small dispersions not only in the position ($\Delta_*x$) and momentum ($\Delta_*p$) but also in energy. The energy dispersion $\Delta_*E$ is defined similarly to (\ref{EqStanddevX}) and (\ref{EqStanddevP}):

\begin{equation}\label{EqStanddevE}\Delta_* E^2=\sum_{n=0}^{\infty}
(E_n-E^*)^2|b_n|^2.\end{equation}
Here $\{E_n\}_{n=0}^\infty$ are the energy eigenvalues and $\{b_n\}_{n=0}^\infty$ are the coefficients in the expansion of the
wave function $\psi$ in the energy eigenfunctions $\{\psi_n\}_{n=0}^\infty$:
$$\psi(x)=\sum_{n=0}^\infty b_n\psi_n.$$

For a particle on the line the energy dispersion is small whenever the momentum dispersion is small, because the momentum and energy operators commute and the energy is a function of momentum: 
\begin{equation}\label{EqHp}
\hat H=\frac{\hat p}{2m}.
\end{equation}
As we pointed out in Subsection~\ref{SecModelDescr}, in bounded domains this is not always the case. For the Hamiltonian  $\hat H_2$ (a particle on a circle), the momentum and energy operators still commute, relation (\ref{EqHp}) holds; therefore, the smallness of $\Delta_*p$ implies the smallness of $\Delta_*E$.

However, for the Hamiltonian $\hat H_1$ (a particle in the infinite well), the position and momentum operators do not commute. Counterintuitively, relation (\ref{EqHp}) does not hold. Therefore, the energy dispersion should be analysed separately.

Let us expand the wave function $\psi(x)$ in the position representation in terms of the momentum eigenfunctions and in terms of the energy eigenfunctions:
\begin{align}\label{EqImpSer}
\psi(x)&=\frac{1}{\sqrt{2l}}\sum_{k=-\infty}^{+\infty}a_k\,e^{i\frac{\pi}{l}kx}\\\label{EqEnSer}&=
\frac{1}{\sqrt l}\sum_{n=1}^\infty b_n\sin\left(\frac{\pi n}{2l}(x-l)\right).
\end{align}
Then, using the expression
$$\frac{1}{\sqrt2l}
\int_{-l}^l\sin\left(\frac{\pi n}{2l}(x-l)\right)\,e^{i\frac{\pi}{l}kx}\,dx=\begin{cases}
\pm(-1)^{\frac{n}{2}}\frac{i}{\sqrt2}&\text{for even $n$ and $k=\pm\frac{n}{2}$,}\\
0&\text{for even $n$ and $k\neq\pm\frac{n}{2}$,}\\
\frac{(-1)^k}{\sqrt2\pi}\frac{n}{k^2-(\frac{n}{2})^2}&\text{for odd $n$}\end{cases}$$
we express the coefficients $\{b_n\}$ in terms of the coefficients $\{a_k\}$:
\begin{equation*}
b_n=\begin{cases}(-1)^{\frac{n}{2}}\frac{i}{\sqrt2}\,(a_{\frac{n}{2}}-a_{-\frac{n}{2}})
&\text{for even $n$},\\
\frac{n}{\sqrt2\pi}\sum\limits_{k=-\infty}^{+\infty}
\frac{(-1)^ka_k}{k^2-(\frac{n}{2})^2}
&\text{for odd $n$.}\end{cases}
\end{equation*}

\begin{proposition}
If the momentum dispersion of a quantum particle in the infinite well in a state $\psi\in L_2(-l,l)$ is finite, then $\psi(l)=\psi(-l)$ (i.e., the boundary condition included in the domain of the operator $\hat p$ is satisfied).
\end{proposition}

\begin{proof}

The convergence of series (\ref{EqStanddevP}) implies the convergence of the series

\begin{equation*}%\label{EqImpOtrInf}
\sum_{k=-\infty}^{+\infty}k^2|a_k|^2<\infty.
\end{equation*}
By Lemma~\ref{LemSerSqrt} (see below), the convergence of this series implies the convergence of the series $\sum_{k=-\infty}^{+\infty}|a_k|$. Then, according to Weierstrass criterion, the series in the Fourier expansions (\ref{EqImpSer}) converge to the function $\psi$ not only in the mean square sense but also absolutely and uniformly. Then, substituting the values $x=\pm l$ into (\ref{EqImpSer}), we obtain the required boundary condition.
\end{proof}

\begin{proposition}
If the energy dispersion of a quantum particle in the infinite well in a state $\psi\in L_2(-l,l)$ is finite, then $\psi(l)=\psi(-l)=0$ (i.e., the boundary conditions included in the domain of the operator $\hat H_1$ is satisfied).
\end{proposition}

\begin{proof}
The proof is similar. The convergence of series (\ref{EqStanddevE}) implies the convergence of the series

\begin{equation*}
\sum_{n=1}^\infty n^4|b_n|^2<\infty.
\end{equation*}
Again, by Lemma~\ref{LemSerSqrt}, the convergence of this series  implies the convergence of the series $\sum_{n=1}^{\infty}|b_n|$. The series in the Fourier expansions (\ref{EqEnSer}) converge to the function $\psi$ absolutely and uniformly. Then, substituting the values $x=\pm l$ into (\ref{EqEnSer}), we obtain the required boundary conditions.
\end{proof}

\begin{lemma}\label{LemSerSqrt}
Let $\sum_{n=1}^{\infty}c_n^2$ be a convergent number series. Then the series
\begin{equation}\label{EqSerSqrt}
\sum_{n=1}^{\infty}\frac{|c_n|}{n}\end{equation}
is also convergent.
\end{lemma}
\begin{proof}
Indeed, since the geometric mean of two numbers is no greater than the arithmetic
mean of these numbers, it follows that
$$\sum_{n=1}^{\infty}\frac{|c_n|}{n}\leq \frac{1}{2}\sum_{n=1}^{\infty}c_n^2
+\frac{1}{2}\sum_{n=1}^{\infty}\frac{1}{n^2}.$$
Both series on the right-hand side converge.
\end{proof}

In the next subsection we will see that these propositions are particular cases of a general relation between the finiteness of dispersion and the domain of the corresponding self-adjoint operator.

We can see that all three families of quantum wave packets satisfy the boundary condition $\psi_\alpha(l)=\psi_\alpha(-l)$, but the only squeezed states given by (\ref{EqGaussCut}) satisfy the condition $\psi_\alpha(\pm l)=0$. Since the theta function $\theta(x,\tau)$ (see (\ref{EqTheta})) has not zeros with real $x$ and real $\tau\neq0$, the condition $\psi_\alpha(\pm l)=0$ cannot be satisfied by squeezed states given by (\ref{EqPsiTheta}). The functions $\psi_\alpha$ given by (\ref{EqPsin}) also cannot satisfy this condition for all $x^*$ and $p^*$.

Hence, the quantum wave packets states $\psi_\alpha$ given by formulae (\ref{EqPsiTheta}) or (\ref{EqPsin})  correspond to infinite energy dispersion and, by this reason, are not satisfactory for the infinite well. But one can suggest the following their improvement. Consider a family states $\psi^{(2l)}_{x^*,p^*,\alpha}$, $x^*\in(-2l,2l)$, $p^*\in\mathbb R$, $\alpha>0$, given by formulae (\ref{EqPsiTheta}) or (\ref{EqPsin}) for even $\varphi(q)$ with $l$ replaced by $2l$. Other words, this is a family of quantum wave packets for the interval $[-2l,2l]$. Now consider the family of states on $L_2(-l,l)$ given by the formula
$$\Psi_{x^*,p^*,\alpha}(x)=\psi_{x^*+l,p^*,\alpha}^{(2l)}(x+l)-\psi_{-x^*-l,-p^*,\alpha}^{(2l)}(-x-l).$$
From formulae (\ref{EqPsiTheta}) or (\ref{EqPsin}) for even $\varphi(q)$, it can be shown that $\Psi_{x^*,p^*,\alpha}(\pm l)=0$. Hence, these states have finite energy dispersion (see the next subsection) and can be used as squeezed states in the infinite well.

\subsection{Domain of a self-adjoint operator and finiteness of dispersion of the physical quantity}\label{SecDisp}

One may notice that whether the dispersion of a physical quantity for a certain state is finite or infinite depends on whether or not this state belongs to the domain of the operator of this physical quantity. Let us show that the following general fact is true: the dispersion of an arbitrary physical quantity (self-adjoint operator) $\hat A$ in a certain state is finite if and only if this state belongs to the domain of $\hat A$. Namely, we prove the following theorem.

\begin{theorem}\label{ThDisp}
Let $A$ be a self-adjoint operator in some Hilbert space $\mathcal H$ and let a quantum system be in a state $\psi\in\mathcal H$. Then the physical quantity corresponding to the operator $A$ has a finite dispersion if and only if $\psi$ belongs to the domain of $A$.
\end{theorem}
\begin{proof}
Represent $A$ by the spectral decomposition \cite{ReedSimon1}:
$$A=\int_{-\infty}^{+\infty}\lambda\,dP_\lambda,$$
where $dP_\lambda$ is a projector-valued measure. The domain of $A$ can be expressed as
$$D(A)=\{\psi|\int_{-\infty}^{+\infty}\lambda^2\,d(\psi,P_\lambda\psi)<\infty\}.$$
The dispersion of the observable corresponding to $A$ for an arbitrary state $\psi$ is
$$\Delta A=\int_{-\infty}^{+\infty}\lambda^2\,d(\psi,P_\lambda\psi)-\overline A^2,$$
where
$$\overline A=\int_{-\infty}^{+\infty}\lambda\,d(\psi,P_\lambda\psi).$$
Obviously, the condition $\psi\in D(A)$ implies $\Delta A<\infty$ and vice versa.
\end{proof}

Physically, the established relation is not  obvious: the ``mathematical'' questions concerning the domains of self-adjoint operators are often omitted in physical literature on quantum mechanics. Here we establish the relation between the finiteness of the dispersion of some physical quantity and the domain of the corresponding self-adjoint operator. This relation can be regarded as a physical meaning of the domain of a self-adjoint operator. 

\subsection{The limit of large interval length and the semiclassical limit}\label{SecLim}
Let us pass to the limit as $l\to\infty$. We will follow the general construction of quantum wave packets given in Subsection~\ref{SecArbitrary}. According to formulae
 (\ref{EqAk}) and (\ref{EqPsin})

$$a_k^{(l)}=
\left[\int_{k-\frac{1}{2}}^
{k+\frac{1}{2}} \varphi_l(q)\,dq\right]^{\frac{1}{2}},$$
$$\psi_l(x)=
\frac{1}{\sqrt{2l}}\sum_{k=-\infty}^{+\infty}a^{(l)}_k\,e^{i\frac{\pi}{l}k(x-x^*)}.$$
Here we supplement $a_k^{(l)}$ and $\psi_l$ with the index $l$ rather than $\alpha$ (which was used before) because now $\alpha$ is a fixed parameter while $l$ varies. Without loss of generality, we assume that $\alpha=1$ because a fixed parameter $\alpha$ can be included in the function $\varphi$.

\begin{theorem}
$$\lim_{l\to\infty}\psi_l(x)=\psi(x)\equiv
\frac{1}{\sqrt{2\pi}}\int_{-\infty}^{+\infty}\sqrt{\varphi(q)}\,
e^{iq(x-x^*)}\,dq.$$
The limit is understood in the pointwise sense.
\end{theorem}

Thus, as $l\to\infty$, the quantum state on the interval constructed by means of a ``discretized'' momentum distribution tends to a state on the real line with the corresponding continuous momentum distribution.

\begin{proof}Firstly, we give a heuristic proof.  By (\ref{EqPhiL}), we have
\begin{equation*}\begin{split}
\psi_l(x)&=
\frac{1}{\sqrt{2l}}\sum_{k=-\infty}^{+\infty}\left\lbrace
\int_{\frac{\pi}{l}(k-\frac{1}{2})}^{\frac{\pi}{l}(k+\frac{1}{2})}
\varphi(q)\,dq\right\rbrace^{\frac{1}{2}}\,e^{i\frac{\pi}{l}k(x-x^*)}\\&=
\frac{1}{\sqrt{2\pi}}
\sum_{k=-\infty}^{+\infty}\sqrt{\varphi(\varkappa^{(l)}_k)}\,\frac{\pi}{l}
\,e^{i\frac{\pi}{l}k(x-x^*)}\\&\to
\frac{1}{\sqrt{2\pi}}\int_{-\infty}^{+\infty}\sqrt{\varphi(q)}\,
e^{iq(x-x^*)}\,dq,
\end{split}\end{equation*}
where $\varkappa_k^{(l)}\in[\frac{\pi}{l}(k-\frac{1}{2}),\frac{\pi}{l}(k+\frac{1}{2})]$. 

But, of course, during this calculations, we interchanged the limits:
\begin{multline*}\lim_{l\to\infty}
\sum_{k=-\infty}^{+\infty}\sqrt{\varphi(\varkappa^{(l)}_k)}\,\frac{\pi}{l}
\,e^{i\frac{\pi}{l}k(x-x^*)}=\lim_{l\to\infty}\lim_{K\to\infty}
\sum_{k=-K}^{K}\sqrt{\varphi(\varkappa^{(l)}_k)}\,\frac{\pi}{l}
\,e^{i\frac{\pi}{l}k(x-x^*)}\\=
\lim_{K\to\infty}\lim_{l\to\infty}
\sum_{k=-K}^{K}\sqrt{\varphi(\varkappa^{(l)}_k)}\,\frac{\pi}{l}
\,e^{i\frac{\pi}{l}k(x-x^*)}.
\end{multline*}
If this interchanging the limits is legitimate, we can perform the passage to the limit as $l\to\infty$ and obtain
$$\lim_{K\to\infty}\lim_{l\to\infty}
\sum_{k=-K}^{K}\sqrt{\varphi(\varkappa^{(l)}_k)}\frac{\pi}{l}
e^{i\frac{\pi}{l}k(x-x^*)}=
\lim_{K\to\infty}\int_{-K}^{K}\sqrt{\varphi(q)}
e^{iq(x-x^*)}dq=\int_{-\infty}^{+\infty}\sqrt{\varphi(q)}
e^{iq(x-x^*)}dq.$$

To justify the interchanging the limits, it suffices to show that the following double limit exists:
\begin{equation}\label{EqLimLK}
\lim_{
\begin{smallmatrix}
l\to\infty\\
K\to\infty
\end{smallmatrix}
}
\frac{1}{\sqrt{2l}}\sum_{k=-Kl}^{Kl}\left\lbrace
\int_{\frac{\pi}{l}(k-\frac{1}{2})}^{\frac{\pi}{l}(k+\frac{1}{2})}
\varphi(q)\,dq\right\rbrace^{\frac{1}{2}}e^{i\frac{\pi}{l}k(x-x^*)}=
\frac{1}{\sqrt{2\pi}}\int_{-\infty}^{+\infty}\sqrt{\varphi(q)}
e^{iq(x-x^*)}dq.
\end{equation}
Let us prove this. By condition, density $\varphi(k)$ has a finite second moment, i.e.,
$\int\limits_{-\infty}^{+\infty}q^2\varphi(q)dq<\infty$. According to assertion
(2a) of Lemma~\ref{LemDeltaP}, this implies that the series
$$\sum_{k=-\infty}^{+\infty}\frac{k^2}{l^2}\int_{\frac{\pi}{l}(k-\frac{1}{2})}^{\frac{\pi}{l}(k+\frac{1}{2})}
\varphi(q)dq$$
converges uniformly in $l\in[l_0,\infty)$, where $l_0>0$ is arbitrary. Then, in the same way as in the proof of Lemma~\ref{LemSerSqrt}, we conclude that the series
$$\sum_{k=-\infty}^{+\infty}\left\lbrace\frac{1}{l}
\int_{\frac{\pi}{l}(k-\frac{1}{2})}^{\frac{\pi}{l}(k+\frac{1}{2})}
\varphi(q)\,dq\right\rbrace^{\frac{1}{2}}$$
as well as the series 
$$\quad
\sum_{k=-\infty}^{+\infty}\left\lbrace\frac{1}{l}
\int_{\frac{\pi}{l}(k-\frac{1}{2})}^{\frac{\pi}{l}(k+\frac{1}{2})}
\varphi(q)\,dq\right\rbrace^{\frac{1}{2}}e^{i\frac\pi lk(x-x^*)}$$
converge uniformly in $l\in[l_0,\infty)$. Then, for any
$\varepsilon>0$, there exists a number $K_0$ such that
\begin{equation}\label{Eql1}
\frac{1}{\sqrt{l}}\left(\sum_{k=-\infty}^{-K_0l}+\sum_{k=K_0l}^{+\infty}\right)
\left\lbrace
\int_{\frac{\pi}{l}(k-\frac{1}{2})}^{\frac{\pi}{l}(k+\frac{1}{2})}
\varphi(q)\,dq\right\rbrace^{\frac{1}{2}}e^{i\frac\pi lk(x-x^*)}<\varepsilon
\end{equation}
for any $l\in[l_0,\infty)$. Let us require that $K_0$ be so large that
\begin{equation}\label{Eql2}\frac{1}{\sqrt{2\pi}}\left(\int_{-\infty}^{-K_0}+\int_{K_0}^{+\infty}\right)
\sqrt{\varphi(q)}\,dq<\varepsilon
\end{equation}
(the integral of $\sqrt{\varphi(q)}$ converges at infinity because
$\sqrt{\varphi(q)}\leq\frac{1}{2}(q^2\varphi(q)+\frac{1}{q^2})$ and the integrals of both functions on the right-hand side converge at infinity).

For a fixed $K_0$ we have
\begin{equation*}
\lim_{l\to\infty}\frac{1}{\sqrt{2l}}\sum_{k=-K_0l}^{K_0l}
\left\lbrace
\int_{\frac{\pi}{l}(k-\frac{1}{2})}^{\frac{\pi}{l}(k+\frac{1}{2})}
\varphi(q)dq\right\rbrace^{\frac{1}{2}}e^{i\frac\pi lk(x-x^*)}
=\frac{1}{\sqrt{2\pi}}\int_{-K_0}^{K_0}\sqrt{\varphi(q)}dq
e^{i\frac\pi lk(x-x^*)}.
\end{equation*}
In this relation, the passage to the limit is legitimate because here we deal with ordinary integral sums on a finite interval. Other words, for any $\varepsilon>0$, there exists $L$ such that
\begin{equation}\label{Eql3}
\left|\frac{1}{\sqrt{2l}}\sum_{k=-K_0l}^{K_0l}
\left\lbrace
\int_{\frac{\pi}{l}(k-\frac{1}{2})}^{\frac{\pi}{l}(k+\frac{1}{2})}
\varphi(q)dq\right\rbrace^{\frac{1}{2}}e^{i\frac\pi lk(x-x^*)}
-\frac{1}{\sqrt{2\pi}}\int_{-K_0}^{K_0}\sqrt{\varphi(q)}e^{i\frac\pi lk(x-x^*)}dq\right|<
\varepsilon.
\end{equation}
for all $l>L$.

From (\ref{Eql1}), (\ref{Eql2}), and (\ref{Eql3}), we conclude that, for any $\varepsilon>0$, there exist  $K_0$ and $L$ such that we have
$$\left|
\frac{1}{\sqrt{2l}}\sum_{k=-Kl}^{Kl}\left\lbrace
\int_{\frac{\pi}{l}(k-\frac{1}{2})}^{\frac{\pi}{l}(k+\frac{1}{2})}
\varphi(q)\,dq\right\rbrace^{\frac{1}{2}}e^{i\frac{\pi}{l}k(x-x^*)}-
\frac{1}{\sqrt{2\pi}}\int_{-\infty}^{+\infty}\sqrt{\varphi(q)}
e^{iq(x-x^*)}dq\right|<3\varepsilon
$$
for all $K>K_0$ and $l>L$. We get formula (\ref{EqLimLK}) which completes the proof of the theorem. \end{proof}

Thus, we obtain a wave packet on the line. In particular, if $\varphi(k)=\frac{1}{\sqrt{2\pi}}e^{-\frac{k^2}{2}}$, then the limit wave packet on the line is Gaussian.

Similar arguments, with slight modifications, can be carried over to squeezed states constructed
by means of the theta function (see Subsection~\ref{SecPsiTheta}). For squeezed states given by truncated Gaussian functions (see Subsection~\ref{SecPsiGauss}), the assertion of the theorem is obvious by construction. Thus, the last property mentioned in the statement of the problem (see the end of Subsection~\ref{SecProblemFormulation})
also holds for the constructed states.

Now, let $\hbar\to0$. Simultaneously, let $\alpha\to\infty$ so that $\hbar\alpha\to0$. Then the formulae proved above (see Theorems~\ref{TheoPsiGauss}, \ref{TheoPsiTheta}, and \ref{TheoPackOtrExist}) for all of three families of wave packets imply that $\overline x_{\alpha}\to x^*$, $\overline p_{\alpha}\to p^*$, $\Delta x_{\alpha}\to0$, $\Delta p_{\alpha}\to0$; i.e., in the semiclassical limit we obtain a point-like particle with prescribed position and momentum.

\section{Conclusions}
We constructed a family of squeezed quantum states on an interval based on the theta function, a family of such states based on truncated Gaussian functions, and a family of quantum wave packets based on the discretization of an arbitrary continuous momentum probability distribution. Estimates on position and momentum dispersion was obtained.

By means of these states, we showed that proper localization of quantum particles in nanoscale space domains is possible. Namely, we saw that, on an interval of order 100~nm, there exist wave packets with a standard deviation of the position of order 0.1~nm and a standard deviation of the momentum of order  $10^{-24}$~kg$\cdot$ m/s. Also the constructed states have finite energy dispersions.

As a supplementary general result, we showed that an arbitrary physical quantity has a finite dispersion if and only if the wave function of a quantum system belongs to the domain of the corresponding self-adjoint operator. This can be regarded as a physical meaning of the domain of a self-adjoint operator.

A continuation of this work is performed in \cite{TrVol-Izv}, where we consider the dynamics of the constructed states (those based on the theta function) on a circle and in the infinite well.

\section*{Acknowledgements}
The authors are grateful for useful remarks and discussions to B.\,L.~Voronov, S.\,Yu.~Dobrokhotov, V.\,I.~Man'ko, A.\,G.~Sergeev, O.\,G.~Smolyanov, A.\,D.~Sukhanov, and E.\,I.~Zelenov. This work was partially supported by the Russian Foundation for Basic Research (project 11-01-00828-a), the Russian Federation's President Programme for the Support of Leading Scientific Schools (project NSh-2928.2012.1), and the Programme of the Division of Mathematics of the Russian Academy of Sciences.

\bigskip\smallskip
\noindent {\Large\textbf{\begin{center}Appendix\end{center}}}

\appendix

\numberwithin{equation}{section}
\numberwithin{lemma}{section}

\section{Asymptotic relations for Gaussian integrals}\label{SecAsympGauss}

The estimate
\begin{equation}\label{EqErfcAsimptMain}
\int_x^{\infty}e^{-\gamma t^2}\,dt=O\left(\frac{e^{-\gamma x^2}}{x}\right),\quad x\to
\infty.
\end{equation}
is well-known.

\begin{lemma}
The following asymptotic formula holds:
\begin{equation}\label{EqErfcDerivAsimptMain}
\int_x^{\infty}t^2e^{-\gamma t^2}\,dt=O(xe^{-\gamma x^2}),\quad
x\to \infty,
\end{equation}
\end{lemma}
\begin{proof}
Let us differentiate the function
$$\Phi(\sqrt\gamma x)=\int_{\sqrt\gamma\, x}^{\infty}e^{-
t^2}\,dt= \sqrt{\gamma}\int_x^{\infty}e^{-\gamma
t^2}\,dt$$ 
with respect to the parameter $\gamma$ at the point $\gamma=1$. On the one hand,
$$
\left.\frac{\partial\Phi(\sqrt\gamma
x)}{\partial\gamma}\right|_{\gamma=1}=\frac{1}{2}\int_x^\infty
e^{-t^2}\,dt-\int_x^\infty t^2e^{-t^2}\,dt.$$ 
On the other hand,
$$\left.\frac{\partial\Phi(\sqrt\gamma
x)}{\partial\gamma}\right|_{\gamma=1}=\frac{x}{2}\,\Phi'(x).$$
Hence, 
\begin{equation}\label{EqErfcDerivAsimptLem}\int_x^\infty
t^2e^{-t^2}\,dt=\frac{1}{2}[\Phi(x)-x\Phi'(x)].
\end{equation}
The required asymptotic formula follows from the fact that $\Phi(x)=O(\frac{e^{-x^2}}{x})$ and $\Phi'(x)=-e^{-x^2}$.
\end{proof}

\section{Asymptotic formulae related to the theta function}\label{SecTheta}

Let us adopt the following (convenient for us) definition of the theta function:
\begin{equation}\label{EqTheta}
\theta(x,\tau)=\sum_{k=-\infty}^{+\infty}e^{-\pi\tau k^2+2\pi ikx},
\end{equation}
where $x$ and $\tau$ are complex numbers with $\Re\tau>0$.

Using the modular property (the Jacobi identity) for the theta function \cite{Karatsuba,Mam}

\begin{equation}\label{EqThetaModular}
\theta\left(\frac{x}{i\tau},\frac{1}{\tau}\right)=\sqrt\tau e^{\frac{\pi x^2}{\tau}}\theta(x,\tau).
\end{equation}
one can prove a number of useful estimates.

\begin{lemma}\label{LemTheta}
The following asymptotic relations hold for an arbitrary real $x$ and for $|a|<\frac{1}{2}$ as $\tau\to0$:
\begin{align}\label{EqThetaVal}
&\theta(x,\tau)=\frac{1}{\sqrt\tau}e^{-\frac{\pi d(x)^2}{\tau}}+
O\left(\frac{1}{\sqrt\tau}e^{-\frac{\pi(1-d(x))^2}{\tau}}\right),\\
\label{EqThetaStdDevK}
&\sum_{k=-\infty}^{+\infty}k^2e^{-\pi\tau k^2}=\frac{1}{2\pi\tau^{3/2}}
+O\left(\frac{1}{2\pi\tau^{3/2}}e^{-\frac{\pi}{\tau}}\right),\\
\label{EqThetaMeanX}
&\int_{-\frac{1}{2}-a}^{\frac{1}{2}-a} x\,|\theta(x,\tau)|^2\,dx=
O\left(e^{-\frac{2\pi}{\tau}(\frac{1}{2}-|a|)^2}\right),\\
\label{EqThetaStdDevX}
&\int_{-\frac{1}{2}-a}^{\frac{1}{2}-a} x^2|\theta(x,\tau)|^2\,dx=\frac{1}{4\pi}\sqrt\frac{\tau}{2}
+O\left(e^{-\frac{2\pi}{\tau}(\frac{1}{2}-|a|)^2}\right).
\end{align}
Here $0\leq d(x)\leq\frac{1}{2}$ is the distance on the real line from $x$ to the nearest integer.
\end{lemma}

\begin{proof}
Using the modular property (\ref{EqThetaModular}), we obtain
$$\theta(x,\tau)=\frac{1}{\sqrt\tau} e^{-\frac{\pi x^2}{\tau}}\theta\left(\frac{x}{i\tau},\frac{1}{\tau}\right)=
\frac{1}{\sqrt\tau}
\sum_{k=-\infty}^{+\infty}e^{-\frac{\pi (k-x)^2}{\tau}}=
\frac{1}{\sqrt\tau}e^{-\frac{\pi d(x)^2}{\tau}}+
O\left(\frac{1}{\sqrt\tau}e^{-\frac{\pi(1-d(x))^2}{\tau}}\right).$$
Estimate (\ref{EqThetaVal}) is proved.
\begin{multline*}
\sum_{k=-\infty}^{+\infty}k^2e^{-\pi\tau k^2}=
-\frac{1}{\pi}\frac{\partial\theta(0,\tau)}{\partial\tau}=
-\frac{1}{\pi}\frac{\partial}{\partial\tau}\left[\frac{1}{\sqrt\tau}
\,\theta\left(0,\frac{1}{\tau}\right)\right]\\=
\frac{1}{2\pi\tau^{\frac{3}{2}}}
\sum_{k=-\infty}^{+\infty}e^{-\frac{\pi k^2}{\tau}}-
\frac{1}{\tau^{\frac{5}{2}}}
\sum_{k=-\infty}^{+\infty}k^2e^{-\frac{\pi k^2}{\tau}}=
\frac{1}{2\pi\tau^{3/2}}
+O\left(\frac{1}{2\pi\tau^{3/2}}e^{-\frac{\pi}{\tau}}\right).\end{multline*}
Estimate (\ref{EqThetaStdDevK}) is proved.

To prove estimate (\ref{EqThetaMeanX}), we use the asymptotic formula (\ref{EqErfcAsimptMain}) for the Gaussian integral and formula (\ref{EqThetaVal}), which is already proved. Then we have
$$
\int_{-\frac{1}{2}-a}^{\frac{1}{2}-a} x\,|\theta(x,\tau)|^2\,dx=
O\left(\frac{1}{\tau}\int_{-\frac{1}{2}-a}^{\frac{1}{2}-a}
xe^{-\frac{2\pi d(x)^2}{\tau}}\,dx\right)$$ 
Let $a\geq0$. Then
$d(x)=|x|$ for $|x|\leq\frac{1}{2}$ and $d(x)=x+1$ for
$x\leq-\frac{1}{2}$. Therefore, in view of (\ref{EqErfcAsimptMain}),
$$\frac{1}{\tau}\int_{-\frac{1}{2}-a}^{\frac{1}{2}-a}
xe^{-\frac{2\pi d(x)^2}{\tau}}\,dx=
\frac{1}{\tau}\int_{-\frac{1}{2}-a}^{-\frac{1}{2}} xe^{-\frac{2\pi
(x+1)^2}{\tau}}\,dx+
\frac{1}{\tau}\int_{-\frac{1}{2}}^{\frac{1}{2}-a} xe^{-\frac{2\pi
x^2}{\tau}}\,dx=O(e^{-\frac{2\pi}{\tau}(\frac{1}{2}-a)^2}),$$
Similarily, if $a\leq0$, then
$$\frac{1}{\tau}\int_{-\frac{1}{2}-a}^{\frac{1}{2}-a}
xe^{-\frac{2\pi d(x)^2}{\tau}}\,dx=
O(e^{-\frac{2\pi}{\tau}(\frac{1}{2}+a)^2}).$$ 
Thus, for an arbitrary $a$ we obtain
$$\int_{-\frac{1}{2}-a}^{\frac{1}{2}-a} x\,|\theta(x,\tau)|^2\,dx=
O(e^{-\frac{2\pi}{\tau}(\frac{1}{2}-|a|)^2}).$$

To prove (\ref{EqThetaStdDevX}), we use asymptotic formulae (\ref{EqErfcDerivAsimptMain}) and (\ref{EqThetaVal}). We have

\begin{multline*}
\int_{-\frac{1}{2}-a}^{\frac{1}{2}-a} x^2|\theta(x,\tau)|^2\,dx=
\frac{1}{\tau}\int_{-\frac{1}{2}-a}^{\frac{1}{2}-a}
x^2e^{-\frac{2\pi d(x)^2}{\tau}}\,dx+ \frac{1}{\tau}\,O\left(
\int_{-\frac{1}{2}-a}^{\frac{1}{2}-a} x^2e^{-\frac{\pi}{\tau}
(d(x)^2+(1-d(x))^2)}\,dx\right)\\=
\frac{1}{\tau}\int_{-\frac{1}{2}-a}^{\frac{1}{2}-a}
x^2e^{-\frac{2\pi d(x)^2}{\tau}}\,dx+ \frac{1}{\tau}\,O\left(
\int_{-\frac{1}{2}-a}^{\frac{1}{2}-a} x^2e^{-\frac{2\pi}{\tau}
(d(x)-\frac{1}{2})^2}\,dx\,e^{-\frac{\pi}{2\tau}}\right).
\end{multline*}
Let $a>0$. Then, in view of (\ref{EqErfcAsimptMain}) and
(\ref{EqErfcDerivAsimptMain}),
\begin{multline*}\frac{1}{\tau}\int_{-\frac{1}{2}-a}^{\frac{1}{2}-a}
x^2e^{-\frac{2\pi d(x)^2}{\tau}}\,dx=
\frac{1}{\tau}\int_{-\frac{1}{2}-a}^{-\frac{1}{2}}
x^2e^{-\frac{2\pi (x+1)^2}{\tau}}\,dx+
\frac{1}{\tau}\int_{-\frac{1}{2}}^{\frac{1}{2}-a}
x^2e^{-\frac{2\pi
x^2}{\tau}}\,dx\\=\frac{1}{4\pi}\sqrt{\frac{\tau}{2}}+
O\left(e^{-\frac{2\pi}{\tau}(\frac{1}{2}-a)^2}\right),
\end{multline*}

\begin{multline*}
\frac{1}{\tau}\int_{-\frac{1}{2}-a}^{\frac{1}{2}-a}
x^2e^{-\frac{2\pi}{\tau}
(d(x)-\frac{1}{2})^2}\,dx\,e^{-\frac{\pi}{2\tau}}\\=
\frac{e^{-\frac{\pi}{2\tau}}}{\tau}\left[
\int_{-\frac{1}{2}-a}^{-\frac{1}{2}}x^2e^{-\frac{2\pi}{\tau}(\frac{1}{2}+x)^2}\,dx+
\int_{-\frac{1}{2}}^{\frac{1}{2}-a}x^2e^{-\frac{2\pi}{\tau}(\frac{1}{2}-x)^2}\,dx\right]=
O\left(e^{-\frac{\pi}{2\tau}}\right),
\end{multline*}
$$\int_{-\frac{1}{2}-a}^{\frac{1}{2}-a} x^2|\theta(x,\tau)|^2\,dx=
\frac{1}{4\pi}\sqrt{\frac{\tau}{2}} +
O\left(e^{-\frac{2\pi}{\tau}(\frac{1}{2}-a)^2}\right).$$
Similarily, for the case of an arbitrary $a$ we have
$$\int_{-\frac{1}{2}-a}^{\frac{1}{2}-a} x^2|\theta(x,\tau)|^2\,dx=
\frac{1}{4\pi}\sqrt{\frac{\tau}{2}} +
O\left(e^{-\frac{2\pi}{\tau}(\frac{1}{2}-|a|)^2}\right).$$
Formula (\ref{EqThetaStdDevX}), as well as the whole lemma, is proved.
\end{proof}

\section{Proof of Lemma \ref{LemDeltaP}.}\label{SecLemDeltaP}

1. By formulae (\ref{EqMeanX}) and (\ref{EqTheoAkSym}) we have
$$\overline p_\alpha=\hbar\sum_{k=-\infty}^{+\infty}p_k|a^{(\alpha)}_k|^2=
\frac{\pi}{l}\hbar\sum_{k=-\infty}^{+\infty}(\overline k+k)|a^{(\alpha)}_{\overline k+k}|^2=
\frac{\pi}{l}\hbar \overline k\sum_{k=-\infty}^{+\infty}|a^{(\alpha)}_k|^2=\frac{\pi}{l}\hbar \overline k.$$
Assertion (1) of the lemma is proved.

2. Let us prove assertion (2a). Without loss of generality, we assume that $p^*=0$. Then

\begin{align*}\Delta_* p_\alpha^2=\sum_{k=-\infty}^{+\infty}
p_k^2|a^{(\alpha)}_k|^2&=
\frac{1}{\alpha}\sum_{k=-\infty}^{+\infty}\int_{k-\frac{1}{2}}^{k+\frac{1}{2}}
\left(\frac{\pi}{l}\hbar k\right)^2\varphi_l\left(\frac{q}{\alpha}\right)\,dq,\\
\widetilde\Delta p_\alpha^2=
\int_{-\infty}^{+\infty}\left(\frac{\pi}{l}\hbar q\right)^2\varphi_{\alpha l}(q)\,
dq&=\frac{1}{\alpha}\sum_{k=-\infty}^{+\infty}\int_{k-\frac{1}{2}}^{k+\frac{1}{2}}
\left(\frac{\pi}{l}\hbar q\right)^2\varphi_l\left(\frac{q}{\alpha}\right)\,dq.\end{align*}

Let us prove the inequality
\begin{equation}\label{EqLemDeltaPFromBelow}
\int_{k-\frac{1}{2}}^{k+\frac{1}{2}}
(k^2-q^2)\,\varphi_l\left(\frac{q}{\alpha}\right)\,dq\geq\varphi_l\left(\frac{k}{\alpha}\right)\int_{k-\frac{1}{2}}^{k+\frac{1}{2}}
(k^2-q^2)\,dq\end{equation}
for an arbitrary integer $k$. Let, for definiteness, $k>0$ (the case of $k<0$ is considered analogously; for $k=0$, the inequality is obvious because this is a maximum point of the function $\varphi_l$). Then the function  $\varphi_l(\frac{q}{\alpha})$ decreases on the interval $q\in[k-\frac{1}{2},k+\frac{1}{2}]$.
On the interval $q\in[k-\frac{1}{2},k]$, we have $k^2-q^2\geq0$ and
$\varphi_l(\frac{\pi}{l}\frac{q}{\alpha})\geq\varphi_l(\frac{\pi}{l}\frac{ k}{\alpha})$; therefore,

$$\int_{k-\frac{1}{2}}^k
(k^2-q^2)\,\varphi_l\left(\frac{q}{\alpha}\right)\,dq\geq\varphi_l\left(\frac{k}{\alpha}\right)\int_{k-\frac{1}{2}}^k
(k^2-q^2)\,dq.$$
On the interval $[k,k+\frac{1}{2}]$, contrarily, $k^2-q^2\leq0$ and
$\varphi_l(\frac{q}{\alpha})\leq\varphi_l(\frac{ k}{\alpha})$; therefore, again we have
$$\int_k^{k+\frac{1}{2}}
(k^2-q^2)\,\varphi_l\left(\frac{q}{\alpha}\right)\,dq\geq\varphi_l\left(\frac{k}{\alpha}\right)\int_k^{k+\frac{1}{2}}
(k^2-q^2)\,dq,$$
which proves the required inequality.

Thus, we have
\begin{multline*}%\label{EqLemDeltaDownBound}
\Delta_* p_\alpha^2-\widetilde\Delta p_\alpha^2=\frac{\hbar^2}{\alpha}\left(\frac{\pi}{l}\right)^2
\sum_{k=-\infty}^{+\infty}
\int_{k-\frac{1}{2}}^{k+\frac{1}{2}}
(k^2-q^2)\,\varphi_l\left(\frac{q}{\alpha}\right)\,dq\geq\\\geq
\frac{\hbar^2}{\alpha}\left(\frac{\pi}{l}\right)^2
\sum_{k=-\infty}^{+\infty}
\varphi_l\left(\frac{k}{\alpha}\right)\int_{k-\frac{1}{2}}^{k+\frac{1}{2}}
(k^2-q^2)\,dq=
-\frac{\hbar^2}{\alpha}
\left(\frac{\pi}{l}\right)^2\frac{1}{12}\sum_{k=-\infty}^{+\infty}
\varphi_l\left(\frac{k}{\alpha}\right)\\to
-\left(\frac{\pi}{l}\hbar\right)^2\frac{1}{12}\int_{-\infty}^{+\infty}
\varphi_l(k)\,dk=-\frac{1}{12}\left(\frac{\pi}{l}\hbar\right)^2,\quad \alpha\to\infty.\end{multline*}
Let us justify the passage to the limit and estimate the rate of convergence of the series to the
integral. By the mean value theorem,
$$\min_{q\in[k-\frac{1}{2},k+\frac{1}{2}]}
\varphi_l\left(\frac{q}{\alpha}\right)\leq\int_{k-\frac{1}{2}}^{k+\frac{1}{2}}\varphi_l
\left(\frac{q}{\alpha}\right)\,dq\leq
\max_{q\in[k-\frac{1}{2},k+\frac{1}{2}]}
\varphi_l\left(\frac{q}{\alpha}\right);$$ therefore
$$\left|\varphi_l
\left(\frac{k}{\alpha}\right)-\int_{k-\frac{1}{2}}^{k+\frac{1}{2}}\varphi_l
\left(\frac{q}{\alpha}\right)\,dq\right|\leq
\max_{q\in[k-\frac{1}{2},k+\frac{1}{2}]}
\varphi_l\left(\frac{q}{\alpha}\right)-\min_{q\in[k-\frac{1}{2},k+\frac{1}{2}]}
\varphi_l\left(\frac{q}{\alpha}\right).$$
Since
\begin{align*}\max_{q\in[k-\frac{1}{2},k+\frac{1}{2}]}
\varphi_l\left(\frac{q}{\alpha}\right)&=\varphi_l\left(\frac{ k-\frac{1}{2}\sgn k}{\alpha}\right),\quad k\neq0,\\
\min_{q\in[k-\frac{1}{2},k+\frac{1}{2}]}
\varphi_l\left(\frac{q}{\alpha}\right)&=\varphi_l\left(\frac{ k+\frac{1}{2}\sgn k}{\alpha}\right),\quad k\neq0,\\
\max_{q\in[-\frac{1}{2},\frac{1}{2}]}
\varphi_l\left(\frac{q}{\alpha}\right)&=\varphi_l(0),\quad \min_{q\in[-\frac{1}{2},\frac{1}{2}]}
\varphi_l\left(\frac{q}{\alpha}\right)=\varphi_l\left(\frac{1}{2\alpha}\right)
\end{align*}
it follows that
\begin{multline*}\sum_{k>0}\left[\max_{q\in[k-\frac{1}{2},k+\frac{1}{2}]}
\varphi_l\left(\frac{q}{\alpha}\right)-\min_{q\in[k-\frac{1}{2},k+\frac{1}{2}]}
\varphi_l\left(\frac{q}{\alpha}\right)\right]\\=
\left[\varphi_l\left(\frac{1}{2\alpha}\right)-\varphi_l\left(\frac{3}{2\alpha}\right)\right]
+\left[\varphi_l\left(\frac{3}{2\alpha}\right)-\varphi_l\left(\frac{5}{2\alpha}\right)\right]
+\ldots=
\varphi_l\left(\frac{1}{2\alpha}\right).
\end{multline*}
Similarily,
$$\sum_{k<0}\left[\max_{q\in[k-\frac{1}{2},k+\frac{1}{2}]}
\varphi_l\left(\frac{q}{\alpha}\right)-\min_{q\in[k-\frac{1}{2},k+\frac{1}{2}]}
\varphi_l\left(\frac{q}{\alpha}\right)\right]=
\varphi_l\left(-\frac{1}{2\alpha}\right)=
\varphi_l\left(\frac{1}{2\alpha}\right).$$
We have
\begin{multline}\label{EqLemSumInt}
\left|\frac{\pi}{l}\sum_{k=-\infty}^{+\infty}
\varphi_l\left(\frac{k}{\alpha}\right)-
\int_{-\infty}^{+\infty}
\varphi_l(q)\,dq\right|\leq
\frac{1}{\alpha}\sum_{k=-\infty}^{+\infty}\left|
\varphi_l\left(\frac{k}{\alpha}\right)-
\int_{k-\frac{1}{2}}^{q+\frac{1}{2}}
\varphi_l\left(\frac{q}{\alpha}\right)\,dq\right|\\=
\frac{1}{\alpha}\left[
\varphi_l(0)+\varphi_l\left(\frac{1}{2\alpha}\right)\right]\leq
\frac{2}{\alpha}\varphi_l(0).
\end{multline}
Thus,
$$\Delta_* p_\alpha^2-\widetilde\Delta p_\alpha^2\geq
-\frac{1}{12}\left(\frac{\pi}{l}\hbar\right)^2
\left[1+\frac{2}{\alpha}\varphi_l(0)\right].$$

We have obtained a lower estimate. Let us find an upper estimate. Arguing as when deriving inequality (\ref{EqLemDeltaPFromBelow}), we obtain
\begin{multline*}\Delta_* p_\alpha^2-\widetilde\Delta p_\alpha^2=\frac{\hbar^2}{\alpha}\left(\frac{\pi}{l}\right)^2
\sum_{k=-\infty}^{+\infty} \int_{k-\frac{1}{2}}^{k+\frac{1}{2}}
(k^2-q^2)\,\varphi_l\left(\frac{q}{\alpha}\right)\,dq\leq\\\leq
\frac{\hbar^2}{\alpha}\left(\frac{\pi}{l}\right)^2
\sum_{k\neq0}\sgn k\left\lbrace \max_{q\in[k-\frac{1}{2},k+\frac{1}{2}]}
\varphi_l\left(\frac{q}{\alpha}\right)\int_{k-\frac{1}{2}\sgn k}^{k} (k^2-q^2)\,dq\right.\\+\left. \min_{q\in[k-\frac{1}{2},k+\frac{1}{2}]}
\varphi_l\left(\frac{q}{\alpha}\right)\int_{k}^{k+\frac{1}{2}\sgn k} (k^2-q^2)\,dq\right\rbrace-
\frac{\hbar^2}{12\alpha}\left(\frac{\pi}{l}\right)^2\varphi_l\left(\frac{1}{2\alpha}\right)\\=
\frac{\hbar^2}{\alpha}\left(\frac{\pi}{l}\right)^2\sum_{k\neq0}
\left\lbrace\frac{1}{4}|k|\left[\max_{q\in[k-\frac{1}{2},k+\frac{1}{2}]}
\varphi_l\left(\frac{q}{\alpha}\right)-
\min_{q\in[k-\frac{1}{2},k+\frac{1}{2}]}
\varphi_l\left(\frac{q}{\alpha}\right)\right]\right.\\-\left.
\frac{1}{24}\left[\min_{q\in[k-\frac{1}{2},k+\frac{1}{2}]}
\varphi_l\left(\frac{q}{\alpha}\right)+
\max_{q\in[k-\frac{1}{2},k+\frac{1}{2}]}
\varphi_l\left(\frac{q}{\alpha}\right)\right]\right\rbrace
-\frac{\hbar^2}{12\alpha}\left(\frac{\pi}{l}\right)^2
\varphi_l\left(\frac{1}{2\alpha}\right).\end{multline*}
Since
\begin{multline}\label{EqLemSumKMaxMin}\frac{1}{\alpha}\sum_{k\neq0}|k|
\left[\max_{q\in[k-\frac{1}{2},k+\frac{1}{2}]}
\varphi_l\left(\frac{q}{\alpha}\right)-
\min_{q\in[k-\frac{1}{2},k+\frac{1}{2}]}
\varphi_l\left(\frac{q}{\alpha}\right)\right]\\=
\left[\varphi_l\left(\frac{1}{2\alpha}\right)-\varphi_l\left(\frac{3}{2\alpha}\right)\right]
+2\left[\varphi_l\left(\frac{3}{2\alpha}\right)-\varphi_l\left(\frac{5}{2\alpha}\right)\right]
+\ldots\\=\varphi_l\left(\frac{1}{2\alpha}\right)+\varphi_l\left(\frac{3}{2\alpha}\right)
+\ldots=
\frac{1}{\alpha}\sum_{k=-\infty}^{+\infty}
\varphi_l\left(\frac{k+\frac{1}{2}}{\alpha}\right),\end{multline}

$$\frac{1}{\alpha}\sum_{k\neq0}\frac{1}{2}\left[\min_{q\in[k-\frac{1}{2},k+\frac{1}{2}]}
\varphi_l\left(\frac{q}{\alpha}\right)+
\max_{q\in[k-\frac{1}{2},k+\frac{1}{2}]}
\varphi_l\left(\frac{q}{\alpha}\right)\right]+\frac{1}{\alpha}
\varphi_l\left(\frac{1}{2\alpha}\right)=\frac{1}{\alpha}\sum_{k=-\infty}^{+\infty}
\varphi_l\left(\frac{k+\frac{1}{2}}{\alpha}\right),$$

\begin{equation}\label{EqLemSumKMaxMin2}
\left|\frac{1}{\alpha}\sum_{k=-\infty}^{+\infty}
\varphi_l\left(\frac{k+\frac{1}{2}}{\alpha}\right)-
\int_{-\infty}^{+\infty}\varphi_l(q)\,dq\right|
\leq\frac{2}{\alpha}\varphi_l(0),\end{equation}
it follows that
$$\Delta p_\alpha^2-\widetilde\Delta p_\alpha^2\leq\frac{1}{6}
\left(\frac{\pi}{l}\hbar\right)^2
\left[1+\frac{2}{\alpha}\varphi_l(0)\right].$$

Returning from the case of $\overline k=0$  to the case of an arbitrary $\overline k$, we should replace $\varphi_l(0)$ by $\varphi_l(0)$. Assertion (2a) of the lemma is proved.

3. Under the conditions of assertion (2b), we can apply the Taylor formula to the function $\varphi_{\alpha l}(k)$:
$$\varphi_l\left(\frac{q}{\alpha}\right)=
\varphi_l\left(\frac{k}{\alpha}\right)+
\frac{1}{\alpha}\varphi_l'\left(\frac{k}{\alpha}\right)
(q-k)+\frac{1}{2\alpha^2}
\varphi_l''\left(\frac{\kappa_k(q)}{\alpha}\right)(q-k)^2,$$
where $\kappa_k(q)\in[q,k]$ if $q\leq k$ and $\kappa_k(q)\in[k,q]$ if $q\geq k$;
$$\Delta\varphi_l\left(\frac{k}{\alpha}\right)
\equiv\varphi_l\left(\frac{k+\frac{1}{2}}{\alpha}\right)-
\varphi_l\left(\frac{k-\frac{1}{2}}{\alpha}\right)=
\frac{1}{\alpha}\varphi_l'\left(\frac{k}{\alpha}\right)+
\frac{1}{8\alpha^2}
\left[\varphi_l''\left(\frac{\kappa^2_k}{\alpha}\right)-
\varphi_l''\left(\frac{\kappa^1_k}{\alpha}\right)\right],$$
where $\kappa^1_k\in[k-\frac{1}{2},k]$ and $\kappa^2_k\in[k,k+\frac{1}{2}]$.

First, let us improve the estimate for the modulus of difference:
$$\left|\varphi_l\left(\frac{k}{\alpha}\right)-
\int_{k-\frac{1}{2}}^{k+\frac{1}{2}}\varphi_l\left(\frac{q}{\alpha}
\right)\,dq\right|\leq\frac{1}{8\alpha^2}
\max_{q\in[k-\frac{1}{2},k+\frac{1}{2}]}
\left|\varphi_l''\left(\frac{q}{\alpha}\right)\right|.$$

Since $\varphi_l''(k)=O(k^{-2})$ as $k\to\infty$, we have
\begin{equation}\label{EqLemSumSecDerivMax}
\frac{1}{\alpha}\sum_{k=-\infty}^{+\infty}
\max_{q\in[k-\frac{1}{2},k+\frac{1}{2}]}
\left|\varphi_l''\left(\frac{q}{\alpha}\right)\right|\to C.\end{equation}
Hence, estimate (\ref{EqLemSumInt}) is also improved:
\begin{equation}\label{EqLemSumInt2}
\left|\frac{1}{\alpha}\sum_{k=-\infty}^{+\infty}
\varphi_l\left(\frac{k}{\alpha}\right)-
\int_{-\infty}^{+\infty}
\varphi_l(q)\,dq\right|\leq
\frac{1}{\alpha}\sum_{k=-\infty}^{+\infty}\left|
\varphi_l\left(\frac{k}{\alpha}\right)-
\int_{k-\frac{1}{2}}^{k+\frac{1}{2}}
\varphi_l\left(\frac{q}{\alpha}\right)\,dq\right|=
O\left(\frac{1}{\alpha^2}\right).\end{equation}

For the difference $\Delta_* p_\alpha^2-\widetilde\Delta p_\alpha^2$ we have

\begin{multline}\label{EqLemDeltaUpBnd}
\Delta_* p_\alpha^2-\widetilde\Delta
p_\alpha^2=\hbar^2\left(\frac{\pi}{l}\right)^2
\sum_{k=-\infty}^{+\infty} \int_{k-\frac{1}{2}}^{k+\frac{1}{2}}
(k^2-q^2)\frac{1}{\alpha} \varphi_l\left(\frac{q}{\alpha}\right)\,dq\\= \hbar^2\left(\frac{\pi}{l}\right)^2
\sum_{k=-\infty}^{+\infty} \int_{k-\frac{1}{2}}^{k+\frac{1}{2}}
(k^2-q^2)\,\left\lbrace\frac{1}{\alpha}\varphi_l\left(\frac{k}{\alpha}\right)+
\frac{1}{\alpha}\Delta\varphi_l\left(\frac{k}{\alpha}\right)(q-k)\right.+\\\left.+ \frac{1}{8\alpha^3}
\left[\varphi_l''\left(\frac{\kappa^1_k}{\alpha}\right)-
\varphi_l''\left(\frac{\kappa^2_k}{\alpha}\right)\right](q-k)+
\frac{1}{2\alpha^3}
\varphi_l''\left(\frac{\kappa_k(q)}{\alpha}\right)(q-k)^2\right\rbrace\,dq\leq\\\leq
\hbar^2\left(\frac{\pi}{l}\right)^2\left\lbrace\sum_{k=-\infty}^{+\infty}\left[
-\frac{1}{12\alpha}\varphi_l\left(\frac{k}{\alpha}\right)
-\frac{1}{6\alpha}k\Delta\varphi_l\left(\frac{k}{\alpha}\right)\right]+
A_1-A_2-A_3\right\rbrace,\end{multline}
where
\begin{align*}A_1&=\sum_{k=-\infty}^{+\infty}
\frac{1}{192\alpha^3}|k|
\left[\max_{\kappa\in[k-\frac{1}{2},k+\frac{1}{2}]}
\varphi_l''\left(\frac{\kappa}{\alpha}\right)-
\min_{\kappa\in[k-\frac{1}{2},k+\frac{1}{2}]}
\varphi_l''\left(\frac{\kappa}{\alpha}\right)\right],\\
A_2&=\sum_{k\neq0}
\frac{1}{320\alpha^3}
\left[\max_{\kappa\in[k-\frac{1}{2},k+\frac{1}{2}]}
\varphi_l''\left(\frac{\kappa}{\alpha}\right)+
\min_{\kappa\in[k-\frac{1}{2},k+\frac{1}{2}]}
\varphi_l''\left(\frac{\kappa}{\alpha}\right)\right],\\
A_3&=
\hbar^2\frac{1}{80\alpha^3}
\varphi_l''\left(\frac{1}{2\alpha}\right).\end{align*}

Here we used the estimate
\begin{multline*}
\sum_{k=-\infty}^{+\infty}\int_{k-\frac{1}{2}}^{k+\frac{1}{2}}
(k^2-q^2)\varphi_l''\left(\frac{\kappa_k(q)}{\alpha}\right)(q-k)^2\,dq\leq\\
\leq\sum_{k\neq0}\sgn k
\left[\max_{\kappa\in[k-\frac{1}{2},k+\frac{1}{2}]}
\varphi_l''\left(\frac{\kappa}{\alpha}\right)
\int_{k-\frac{1}{2}\sgn k}^{k}(k^2-q^2)(q-k)^2\,dq\right.\\
\left.+\min_{\kappa\in[k-\frac{1}{2},k+\frac{1}{2}]}
\varphi_l''\left(\frac{\kappa}{\alpha}\right)
\int_{k}^{k+\frac{1}{2}\sgn k}(k^2-q^2)(q-k)^2\,dq\right]-
\frac{1}{80}\,\varphi_l\left(\frac{1}{2\alpha}\right)\\=
\sum_{k\neq0}\left\lbrace\frac{1}{32}\:|k|
\left[\max_{\kappa\in[k-\frac{1}{2},k+\frac{1}{2}]}
\varphi_l''\left(\frac{\kappa}{\alpha}\right)-
\min_{\kappa\in[k-\frac{1}{2},k+\frac{1}{2}]}
\varphi_l''\left(\frac{\kappa}{\alpha}\right)\right]\right.\\\left.
-\frac{1}{160}\left[\max_{\kappa\in[k-\frac{1}{2},k+\frac{1}{2}]}
\varphi_l''\left(\frac{\kappa}{\alpha}\right)+
\min_{\kappa\in[k-\frac{1}{2},k+\frac{1}{2}]}
\varphi_l''\left(\frac{\kappa}{\alpha}\right)\right]\right\rbrace
-\frac{1}{80}\,\varphi_l''\left(\frac{1}{2\alpha}\right).
\end{multline*}
Similarly,
\begin{equation}\label{EqLemDeltaDwnBnd}
\Delta_* p_\alpha^2-\widetilde\Delta
p_\alpha^2\geq\hbar^2\left(\frac{\pi}{l}\right)^2\left\lbrace
\sum_{k=-\infty}^{+\infty}\left[
-\frac{1}{12\alpha}\varphi_l\left(\frac{k}{\alpha}\right)
-\frac{1}{6\alpha}k\Delta\varphi_l\left(\frac{k}{\alpha}\right)\right]
-A_1+A_2+A_3\right\rbrace.
\end{equation}

Taking into account (\ref{EqLemSumInt2}) and (\ref{EqLemSumKMaxMin}), we obtain
\begin{multline}\label{EqLemDeltaMain}
\sum_{k=-\infty}^{+\infty}\left[
-\frac{1}{12\alpha}\varphi_l\left(\frac{k}{\alpha}\right)
-\frac{1}{6\alpha}k\Delta\varphi_l\left(\frac{k}{\alpha}\right)\right]=
\left(-\frac{1}{12}+\frac{1}{6}\right)
\int_{-\infty}^{+\infty}
\varphi_l(k)\,dk+O\left(\frac{1}{\alpha^2}\right)\\=
\frac{1}{12}+O\left(\frac{1}{\alpha^2}\right),\end{multline}
The sign of the second term has changed because, in view of the monotonicity,
$$\Delta\varphi_l\left(\frac{k}{\alpha}\right)=
-\sgn k\left[\max_{q\in[k-\frac{1}{2},k+\frac{1}{2}]}
\varphi_l\left(\frac{q}{\alpha}\right)-
\min_{q\in[k-\frac{1}{2},k+\frac{1}{2}]}
\varphi_l\left(\frac{q}{\alpha}\right)\right]$$
(in particular, $\Delta\varphi_l\left(\frac{k}{\alpha}\right)$ and $k$
have opposite signs).

By calculations similar to (\ref{EqLemSumKMaxMin}) and (\ref{EqLemSumKMaxMin2}),
one can easily show that  $A_1=O(1/\alpha^2)$ (in the present case, calculations are more cumbersome because $\varphi_l''$ may have more than one extremum;
however, one can divide the whole line into segments between the extremum points and consider separately these segments and the extremum points; the number of extremum points is finite by the hypothesis).

 According to (\ref{EqLemSumSecDerivMax}),
$A_2=O(1/\alpha^2)$.
Obviously, $A_3=O(1/\alpha^3)$. Then, taking into account (\ref{EqLemDeltaUpBnd}), (\ref{EqLemDeltaDwnBnd}) and (\ref{EqLemDeltaMain}), we obtain

$$\Delta_\alpha p=\widetilde\Delta_\alpha p+
\left(\frac{\pi}{l}\hbar\right)^2\left[\frac{1}{12}+O\left(\frac{1}{\alpha^2}\right)
\right].$$

The lemma is proved.

\section{Proof of Lemma~\ref{LemSumCos}.}\label{SecLemSumCos}

It follows from the hypothesis of the lemma that there are two possibilities: (1) all
$a_k\geq0$ and $a_0=\max a_k>0$, and (2) all $a_k\leq0$ and $a_0=\min a_k<0$. In any case, $|a_0|=\max|a_k|$. Since one case can be reduced to the other by changing the signs of all $a_k$, which does not effect any side of inequality (\ref{EqLemChiBound}), we assume without loss of generality that the first variant takes place: all $a_k\geq0$.

Take a natural $K$ and consider the first
$K$ elements. $a_0\leq a_1\geq\ldots\geq a_K$, $a_0>0$, in the sequence. We will construct step by step a new subsequence of $(K+1)$ elements. At the zeroth step, we have
$$a^{(0)}_0\equiv a_0,\quad a^{(0)}_1\equiv a_1,\quad a^{(0)}_2\equiv a_2,\quad \ldots,\quad a^{(0)}_K\equiv a_K,\quad a^{(0)}_{K+1}\equiv0.$$
At each step we will keep the sequence monotonic. For any step $m\geq0$, we define two numbers $K_1^{(m)}\geq1$ and
$K_2^{(m)}\geq K_1^{(m)}\geq1$ in the following manner:
$$a^{(m)}_0=a^{(m)}_1=\ldots=a^{(m)}_{K^{(m)}_1-1}\neq a^{(m)}_{K^{(m)}_1},\quad a^{(m)}_{K^{(m)}_1}=a^{(m)}_{K^{(m)}_1+1}=\ldots=a^{(m)}_{K^{(m)}_2}\neq
a^{(m)}_{K^{(m)}_2+1}.$$

The partial sum of the series at step $m$ is
\begin{equation}\label{EqLemSumK}
S^{(m)}_K=a^{(m)}_0\sum_{k=0}^{K^{(m)}_1-1}\cos kx+
a^{(m)}_{K^{(m)}_1}D^{(m)}_K+
\sum_{k=K^{(m)}_2+1}^Ka^{(m)}_k\cos kx,\end{equation}
where
$$D^{(m)}_K=\sum_{k=K^{(m)}_1}^{K^{(m)}_2}a^{(m)}_k\cos kx=
a^{(m)}_{K^{(m)}_1}\sum_{k=K^{(m)}_1}^{K^{(m)}_2}\cos kx.$$

Now, let us directly describe the iteration rule for constructing the sequence $a^{(m+1)}_0,a^{(m+1)}_1,\ldots,a^{(m+1)}_{K+1}$
from $a^{(m)}_0,a^{(m)}_1,\ldots,a^{(m)}_{K+1}$, $m\geq0$.

If $D^{(m)}_K\geq0$, then the terms of the sequence with indices from
$K^{(m)}_1$ to $K^{(m)}_2$ inclusive take the value
$a^{(m)}_0$ (i.e., they increase up to the preceding term of the sequence), other terms remaining
unchanged:
\begin{equation*}%\label{EqLemIter1}
a^{(m+1)}_k=\begin{cases}a^{(m)}_0& \text{for $K^{(m)}_1\leq k\leq K^{(m)}_2$},\\
a^{(m)}_k&\text{otherwise}.\end{cases}
\end{equation*} 
In this case,  $K^{(m+1)}_1=K^{(m)}_2+1$ and $K^{(m+1)}_2\geq
K^{(m)}_2+1$.

If $D^{(m)}_K\leq0$, then the terms of the sequence with indices from
$K^{(m)}_1$ to $K^{(m)}_2$ inclusive take the value
$a^{(m)}_{K^{(m)}_2+1}$ (i.e., they decrease down to the subsequent term of the sequence), other terms  remaining unchanged:
\begin{equation*}%\label{EqLemIter2}
a^{(m+1)}_k=\begin{cases}a^{(m)}_{K^{(m)}_2+1}& \text{for $K^{(m)}_1\leq k\leq K^{(m)}_2$},\\
a^{(m)}_k&\text{otherwise}.\end{cases}
\end{equation*} 
In this case $K^{(m+1)}_1=K^{(m)}_1$ and $K^{(m+1)}_2\geq K^{(m)}_2+1$.

One can notice that in either case the sequence remains monotonic. In addition, $a^{(m)}_0\equiv a_0$ does not change its value, since $K^{(m)}\geq1$. Therefore, $a^{(m)}_0=a_0=\max
a^{(m)}_k$ and $a^{(m)}_k\leq a^{(m)}_{K^{(m)}_2+1}$ for $k\leq
K^{(m)}_2$. Consequently, if $D^{(m)}_K\geq0$, then
$a^{(m+1)}_{K^{(m)}_1}\geq a^{(m)}_{K^{(m)}_1}$, and if
$D^{(m)}_K\leq0$, then $a^{(m+1)}_{K^{(m)}_1}\leq
a^{(m)}_{K^{(m)}_1}$. Then in either case the term
$a^{(m)}_{K^{(m)}_1}D^{(m)}_K$ in sum (\ref{EqLemSumK}) does not decrease, while the other two terms remain unchanged. Thus, the whole partial sum does not decrease: $S^{(m+1)}_K\geq S^{(m)}_K$.

After a certain number $M\leq K$ of steps we obtain $K^{(M)}_2=K+1$, $1\leq K^{(M)}_1\leq
K^{(M)}_2$, and
\begin{equation}\label{EqLemAFin}
a^{(M)}_k=\begin{cases}a^{(0)},& \text{$0\leq k\leq K^{(M)}_1-1$},\\
0,& \text{$k\geq K^{(M)}_1$}\end{cases}
\end{equation} 
(the last line is correct because $a^{(M)}_{K+1}=a^{(0)}_{K+1}=0$). Thus,
$$S_K\equiv S^{(0)}_K\leq a_0\sum_{k=0}^{K^{(M)}_1-1}\cos kx.$$
Using the well-known formula
\begin{equation}
\sum_{n=0}^{K-1}\cos kx=\frac{\sin\frac{Kx}{2}\cos\frac{(K-1)x}{2}} {\sin\frac{x}{2}},
\end{equation}
we obtain
$$S_K\leq a_0\frac{\sin\frac{K^{(M)}_1x}{2}
\cos\frac{(K^{(M)}_1-1)x}{2}}
{\sin\frac{x}{2}}.$$

In a similar way one can deduce a lower estimate for the partial sum. To this end, one should run the iteration process from right to left, rather than from left to right. $K^{(m)}_1$ is the number after which all terms of the sequence are equal to $a_{K+1}$ (i.e., vanish), and $K^{(m)}_2\leq
K^{(m)}_1$ is the number starting from which the terms of the sequence have equal values up to the $K^{(m)}_1$th term. $D^{(m)}_K$ is
defined as before except that the summation limits are interchanged. If $D^{(m)}_K\geq0$, then the values of the terms with numbers from
$K^{(m)}_2$ to $K^{(m)}_1$ decrease down to $a_{K^{(m)}_1+1}=0$,
and if $D^{(m)}_K\leq0$, they increase up to $a_{K^{(m)}_2-1}$. Then the value of the partial sum does not increase. After a finite 2
number of iteration steps, we again obtain formula  (\ref{EqLemAFin}) in which
$K^{(m)}_1$ is replaced by $K^{(m)}_2$ and the estimate
$$S_K\geq a_0\frac{\sin\frac{K^{(M)}_2x}{2}
\cos\frac{(K^{(M)}_2-1)x}{2}}
{\sin\frac{x}{2}}.$$ Combining the upper and lower estimates, we obtain
\begin{equation*}
|S_K|\leq\frac{|a_0|}{|\sin\frac{x}{2}|}.
\end{equation*}
Note that the estimate of the partial sum does not depend on  $K$. This yields estimate
 (\ref{EqLemChiBound}) for the sum of the whole series.


\begin{thebibliography}{99}
\bibitem{Schroe-coher}E.~Schr\"{o}dinger, ``Der stetige \"{U}bergang von der Mikro- zur Makromechanik,'' Naturwissenschaften \textbf{14}, 664--666 (1926).

\bibitem{Neumann}J.~von~Neumann, \textit{Mathematische Grundlagen der Quantenmechanik} (Julius Springer, Berlin, 1932).

\bibitem{Schleich} W.~P.~Schleich, \textit{Quantum Optics in Phase Space} (Wiley-VCH, Weinheim, 2001).

\bibitem{Wei}
A.~Weil, ``Sur certains groupes d'op\'{e}rateurs unitaires,'' Acta Math. \textbf{111}, 143--211 (1964).

\bibitem{CN}P.~Carruthers and M.~M.~Nieto, ``Phase and angle variables in quantum mechanics,'' Rev. Mod. Phys. \textbf{40} (2), 411--440 (1968).

\bibitem{KS}
\textit{Coherent States: Applications in Physics and Mathematical Physics}, Ed. by J.~R.~Klauder and B.-S.~Skagerstam
(World Sci., Singapore, 1985).

\bibitem{Per}
A.~M.~Perelomov, \textit{Generalized Coherent States and Their Applications} (Nauka, Moscow,
1987; Springer, Berlin, 1986).

\bibitem{TAG-book}S.~Twareque Ali, J.-P.~Antoine, and J.-P.~Gazeau, Coherent States, Wavelets and Their Generalizations (Springer-Verlag, New York, 1999).

\bibitem{V}A.~Vourdas, ``Analytic representations in quantum mechanics'', J. Phys. A: Math. Gen. \textbf{39}, R65–R141 (2006).

\bibitem{G-book}J.-P.~Gazeau, \textit{Coherent States in Quantum Physics} (Wiley-VCH, Weinheim, 2009).

\bibitem{CD-book}M.~Combescure and D.~Robert, \textit{Coherent States and Applications in Mathematical Physics} (Springer, Berlin, 2012).


\bibitem{BG}S.~De Bi\`{e}vre S. and J.~A.~Gonz\'{a}lez,
``Semiclassical behaviour of coherent states on the circle'', in
\textit{Quantization and coherent states methods} (World Sci., Singapore, 1993), pp. 152--157.

\bibitem{KRCircle}K.~Kowalski  and J.~Rembieli\'{n}ski and L.~C.~Papaloucas, ``Coherent states for a quantum particle on a circle,'' J. Phys. A: Math. Gen. \textbf{29}, 4149 (1996);  \href{http://arxiv.org/abs/quant-ph/9801029}{arXiv:quant-ph/9801029}.

\bibitem{GdOCircle}J.~A.~Gonz\'{a}lez and M.~A.~del Olmo, ``Coherent states on the circle,'' J. Phys. A: Math. Gen. \textbf{31} (44), 8841--8857 (1998); \href{http://arxiv.org/abs/quant-ph/9809020}{arXiv:quant-ph/9809020}.

\bibitem{HM}B.~C.~Hall and J.~J.~Mitchell, ``Coherent states on spheres,'' J. Math. Phys. \textbf{43} (3) 1211--1236 (2002); \href{http://arxiv.org/abs/quant-ph/0109086}{arXiv:quant-ph/0109086}.


\bibitem{KR2002}
K.~Kowalski  and J.~Rembieli\'{n}ski, ``On the uncertainty relations and
squeezed states for the quantum mechanics on a circle,'' J. Phys. A: Math. Gen. \textbf{35}, 1405–-1414 (2002); \href{http://arxiv.org/abs/quant-ph/0202070}{arXiv:quant-ph/0202070}.

\bibitem{GCircle}I.~Aremua, J.~P.~Gazeau, and M.~N.~Hounkonnou, ``Action-angle coherent states for quantum systems with cylindric phase space,'' J. Phys. A: Math. Theor. \textbf{45} (33) 335302 (2012); \href{http://arxiv.org/abs/1111.4908}{arXiv:1111.4908 [quant-ph]}.

\bibitem{CLTCircle}G.~Chadzitaskos, P.~Luft, and J.~Tolar, ``Quantizations on the circle and coherent states,'' J. Phys. A: Math. Theor. \textbf{45} (24) 244027 (2012); \href{http://arxiv.org/abs/1201.3895}{arXiv:1201.3895 [quant-ph]}.


\bibitem{PTinf}J-P.~Antoine, J-P.~Gazeau, P.~Monceau, J.~R.~Klauder, and K.~A.~Penson, ``Temporally stable coherent states for infinite well and P\"{o}schl-Teller potentials,'' J. Math. Phys. \textbf{42} (6), 2349--2387 (2001); \href{http://arxiv.org/abs/math-ph/0012044}{arXiv:math-ph/0012044}.

\bibitem{GenIntel}A.~H.~EL~Kinani and M.~Daoud, ``Generalized intelligent states for an arbitrary quantum system,'' J. Phys. A: Math. Gen. \textbf{34} (26), 5373--5387 (2001); \href{http://arxiv.org/abs/quant-ph/0311029}{arXiv:quant-ph/0311029}.


\bibitem{Gazeau}P.~L.~Garcia~de~Leon, J.~P.~Gazeau, and J.~Queva, ``Infinite quantum well: a coherent state approach,'' Phys. Lett. A \textbf{372} (20) 3597--3607 (2008); \href{http://arxiv.org/abs/0802.3551}{arXiv:0802.3551 [quant-ph]}.


\bibitem{HW}M.~H.~Al-Hashimi and U.-J.~Wiese, ``From a particle in a box to the uncertainty relation in a quantum dot and to reflecting walls for relativistic fermions'', Ann. Phys. \textbf{327} (1), 1--28 (2012); \href{http://arxiv.org/abs/1105.0391}{arXiv:1105.0391 [quant-ph]}.




\bibitem{nano}K.~E.~Drexler, \textit{Nanosystems: Molecular Machinery, Manufacturing, and Computation} (J. Wiley \& Sons, New
York, 1992).

\bibitem{ReedSimon}M.~Reed and B.~Simon, \textit{Methods of Modern Mathematical Physics, Vol. 2: Fourier Analysis, Self-Adjointness}
(Academic, New York, 1975).

\bibitem{Teach}G.~Bonneau, J.~Faraut, and G.~Valent, ``Self-adjoint extensions of operators and the teaching of quantum mechanics,'' Am. J. Phys., \textbf{69} (3), 322--331 (2001); \href{http://arxiv.org/abs/math-ph/0004012}{arXiv:math-ph/0004012}.

\bibitem{VGT}B.~L.~Voronov, D.~M.~Gitman, and I.~V.~Tyutin, ``Self-adjoint differential operators assosiated with self-adjoint differential expressions'', \href{http://arxiv.org/abs/quant-ph/0603187}{arXiv:quant-ph/0603187}.

\bibitem{VGT1}B.~L.~Voronov, D.~M.~Gitman, and I.~V.~Tyutin, ``Constructing quantum observables and self-adjoint extensions of symmetric operators I'', Russ. Phys. J., \textbf{50} (1), 1--31 (2007).

\bibitem{VGT2}B.~L.~Voronov, D.~M.~Gitman, and I.~V.~Tyutin, ``Constructing quantum observables and self-adjoint extensions of symmetric operators II. Differential operators'', Russ. Phys. J., \textbf{50} (9), 853--884 (2007).

\bibitem{VGT3}B.~L.~Voronov, D.~M.~Gitman, and I.~V.~Tyutin, ``Constructing quantum observables and self-adjoint extensions of symmetric operators II. Self-adjoint boundary conditions'', Russ. Phys. J., \textbf{51} (2), 115--157 (2008).

\bibitem{Naimark}M.~A.~Naimark, \textit{Linear Differential Operators} (Nauka, Moscow, 1969) [in Russian].

\bibitem{Karwowski}P.~Garbaczewski and W.~Karwowski, ``Impenetrable barriers and canonical quantization,'' Am. J. Phys. \textbf{72} (7),
924-–933 (2004); \href{http://arxiv.org/abs/math-ph/0310023}{arXiv:math-ph/0310023}.


\bibitem{Nov} S.~P.~Novikov, ``1. Classical and modern topology.
2. Topological phenomena in real world physics,'' in \textit{Visions in Mathematics} (Birkh\"{a}user, Basel, 2010);
\href{http://arxiv.org/abs/math-ph/0004012}{arXiv:math-ph/0004012}.

\bibitem{Judge}D.~Judge, ``On the uncertainty relation for angle variables,'' Nuovo Cimento \textbf{31} (2), 332–340 (1964).

\bibitem{Davydov}A.~S.~Davydov, \textit{Quantum Mechanics} (Nauka, Moscow, 1973) [in Russian].

\bibitem{SPb}A.~V.~Golovnev, L.~V.~Prokhorov, ``Uncertainty relations in curved spaces,'' J. Phys. A \textbf{37} (7), 2765--2775 (2004); \href{http://arxiv.org/abs/quant-ph/0306080}{arXiv:quant-ph/0306080}.

\bibitem{Trifonov}D.~A.~Trifonov, ``On the position uncertainty measure on the circle,'' J. Phys. A \textbf{36} (47), 11873--11879 (2003); \href{http://arxiv.org/abs/quant-ph/0307137}{arXiv:quant-ph/0307137}.

\bibitem{Dumitru}S.~Dumitru, ``A possible general approach regarding the conformability of angular observables with mathematical rules of quantum mechanics,'' \href{http://arxiv.org/abs/quant-ph/0602147}{arXiv:quant-ph/0602147}.

\bibitem{ReedSimon1}M.~Reed and B.~Simon, \textit{Methods of Modern Mathematical Physics, Vol. 1: Functional analysis}
(Academic, New York, 1980).

\bibitem{Karatsuba}S.~M.~Voronin and A.~A.~Karatsuba, \textit{The Riemann Zeta-Function} (Fizmatlit, Moscow, 1994; de Gruyter, Berlin, 1992).

\bibitem{Mam}
D.~Mumford, Tata Lectures on Theta. I, II (Birkh\"{a}user, Boston, 1983, 1984).

\bibitem{TrVol-Izv}I.~V.~Volovich and A.~S.~Trushechkin, ``Asymptotic properties of quantum dynamics in bounded domains at various time scales,'' Izv. Math. \textbf{76} (1), 39--78 (2012); \href{http://arxiv.org/abs/1304.2332}{arXiv:1304.2332 [quant-ph]}.



\end{thebibliography}
\end{document}